\documentclass[aps, prl, superscriptaddress, reprint, twocolumn]{revtex4-2}
\usepackage[T1]{fontenc}
\usepackage[utf8]{inputenc}

\usepackage{amsmath, amssymb, amstext, amsfonts, amsthm}
\usepackage{mathtools}
\usepackage{dsfont}
\usepackage{graphicx}
\usepackage[dvipsnames]{xcolor}
\usepackage{float}
\usepackage{ragged2e}
\usepackage[caption=false, justification=justified]{subfig}
\DeclareCaptionJustification{justified}{\justifying}
\usepackage{placeins}
\usepackage{standalone}
\usepackage{tikz}
\usepackage{enumitem}
\usepackage{microtype}
\usepackage{times}

\usepackage{subdepth} 
\usepackage[colorlinks=true, linkcolor = BlueViolet, citecolor = BlueViolet, urlcolor = BlueViolet]{hyperref}

\theoremstyle{definition}

\newtheorem{theorem}{Theorem}

\newtheorem{lemma}[theorem]{Lemma}

\newtheorem*{claim*}{Claim}

\theoremstyle{remark}

\newcommand{\R}{\mathbb{R}}
\newcommand{\C}{\mathbb{C}}

\DeclareMathOperator{\re}{Re}
\DeclareMathOperator{\im}{Im}

\DeclareMathOperator{\Tr}{Tr}
\DeclareMathOperator{\sinc}{sinc}
\newcommand{\proj}{\hat{\Pi}}

\newcommand{\idop}{\hat{\mathds{1}}}
\newcommand{\diff}{\mathrm{d}}

\newcommand{\dos}{\mathcal{N}}
\newcommand{\tildos}{\widetilde{\mathcal{N}}}
\newcommand{\kraus}[2]{\hat{\mathcal{K}}_{#1}(#2)}
\newcommand{\krausd}[2]{\hat{\mathcal{K}}_{#1}^\dagger(#2)}
\newcommand{\scr}{\text{scr}}

\newcommand{\doseff}{\mathcal{N}_{\beta}}
\newcommand{\tildoseff}{\widetilde{\mathcal{N}}_{\beta}}
\newcommand{\exprate}{\Gamma}   
\newcommand{\expratemain}{\Lambda}    
\newcommand{\exprateparam}{\lambda}
\newcommand{\tcomp}{\mathbf{t}}
\newcommand{\texcept}{W_{\text{ex}}}
\newcommand{\rexcept}{V}

\newcommand{\cbeta}{c\beta}

\newcommand{\edit}[1]{{#1}} 
\newcommand{\editb}[1]{{#1}}

\newcommand{\maintexteqinitstatedef}{\eqref{eq:initstatedef}}
\newcommand{\maintexteqprojdef}{\eqref{eq:projdef}}
\newcommand{\maintexteqnewscrpsdef}{\eqref{eq:newscrpsdef}}
\newcommand{\maintexteqpurityandps}{\eqref{eq:purity_and_ps}}
\newcommand{\maintexteqscrtimecriterion}{\eqref{eq:scrtimecriterion2}}
\newcommand{\maintexteqfidelityinfdef}{\eqref{eq:fidelityinfdef}}

\newcommand{\maintexteqtsscrboundmain}{\eqref{eq:ts_scr_bound_main}}
\newcommand{\maintexteqentropyscrambling}{\eqref{eq:entropyscrambling}}
\newcommand{\maintexteqtsentropyform}{\eqref{eq:ts_entropy_form}}
\newcommand{\maintexteqtimeevolvingstateKraus}{\eqref{eq:timeevolvingstateKraus}}

\newcommand{\maintexteqdynamicalinequalityagain}{\eqref{eq:dynamicalinequalityagain}}
\newcommand{\maintexteqregSFFdef}{\eqref{eq:regSFFdef}}
\newcommand{\maintexteqtildoseffdef}{\eqref{eq:tildoseffdef}}

\newcommand{\maintexteqdosscramblingcriteriongen}{\eqref{eq:dosscramblingcriterion_gen}}
\newcommand{\maintexteqridgeproperty}{\eqref{eq:ridgeproperty}}

\newcommand{\maintexteqdoubleexpbound}{\eqref{eq:doubleexpbound1}}
\newcommand{\maintexteqexceptionalsetlength}{\eqref{eq:exceptionalsetlength}}
\newcommand{\maintexteqexprateresult}{\eqref{eq:exprateresult}}
\newcommand{\maintexteqtsuniversalbound}{\eqref{eq:ts_universal_bound}}

\newcommand\numberthis{\stepcounter{equation}\tag{\theequation}}

\begin{document}

\title{Proof of a Universal Speed Limit on Fast Scrambling in Quantum Systems}
\author{Amit Vikram}
\affiliation{Joint Quantum Institute and Department of Physics, University of Maryland, College Park, MD 20742, USA}
\author{Laura Shou}
\affiliation{Joint Quantum Institute and Department of Physics, University of Maryland, College Park, MD 20742, USA}
\affiliation{Condensed Matter Theory Center, Department of Physics, University of Maryland, College Park, MD 20742, USA}
\author{Victor Galitski}
\affiliation{Joint Quantum Institute and Department of Physics, University of Maryland, College Park, MD 20742, USA}


\begin{abstract}
We prove that the time required for sustained information scrambling in any Hamiltonian quantum system is universally at least logarithmic in the entanglement entropy of scrambled states. This addresses two foundational problems in nonequilibrium quantum dynamics. (1)~It sets the earliest possible time for the applicability of equilibrium statistical mechanics in a quantum system coupled to a bath at a finite temperature. (2)~It proves a version of the fast scrambling conjecture, originally motivated in models associated with black holes, as a fundamental property of quantum mechanics itself. Our result builds on a refinement of the energy-time uncertainty principle in terms of the infinite temperature spectral form factor in quantum chaos. We generalize this formulation to arbitrary initial states of the bath, including finite temperature states, by mapping Hamiltonian dynamics with any initial state to nonunitary dynamics at infinite temperature. A regularized spectral form factor emerges naturally from this procedure, whose decay is universally constrained by analyticity in complex time. This establishes an exact speed limit on information scrambling by the most general quantum mechanical Hamiltonian, without any restrictions on locality or the nature of interactions.
\end{abstract}

\maketitle


\textit{Introduction}--- Consider a quantum mechanical system with Hilbert space $\mathcal{H}_S$, e.g., a collection of $N_S$ qubits with dimension $D_S = 2^{N_S}$, that is initially in a pure basis state $\lvert k\rangle_S$ (which we call the computational basis). At the time $t=0$, we couple $\mathcal{H}_S$ to an external system (say, of $N_E$ qubits and dimension $D_E = 2^{N_E}$) in a state $\hat{\rho}_{\beta E}$ (\edit{for instance, but not necessarily,} a thermal state with inverse temperature $\beta$), and allow their mutual interactions to drive time evolution in the combined system. The initial state of the overall system $\mathcal{H}_S \otimes \mathcal{H}_E$ of dimension $D = D_S D_E$ is given by the density operator:
\begin{equation}
    \hat{\rho}_k(0) = \lvert k\rangle_S\langle k\rvert \otimes \hat{\rho}_{\beta E}.
    \label{eq:initstatedef}
\end{equation}
Generically, information about the initial state of $\mathcal{H}_S$ is ``scrambled'' by interactions, being lost to $\mathcal{H}_E$ through the generation of entanglement, and not entirely recoverable from $\mathcal{H}_S$ alone. To examine to what extent any basis $\lvert \alpha\rangle_S$ of the subsystem $\mathcal{H}_S$ retains this information, we probe the system with measurements represented by the orthonormal projectors
\begin{equation}
    \proj_{\alpha} = \lvert \alpha\rangle_S\langle \alpha\rvert \otimes \idop_E.
    \label{eq:projdef}
\end{equation}

The above is a minimal description of quantum information scrambling that has far-reaching connections to statistical mechanics. Most directly, it is the prototypical setting for the nonequilibrium process of thermalization~\cite{vonNeumannThermalization, tumulka_CT, CanonicalTypicalityPSW, NormalTypicality, Nandkishore, DAlessio2016, Borgonovi2016, EisertShortReview, MBLThermalization, GongHamazakiNoneqbBounds}, of the system $\mathcal{H}_S$ via interactions with a thermal reservoir $\mathcal{H}_E$, which (if it occurs) is ultimately responsible for the validity of equilibrium statistical mechanics in $\mathcal{H}_S$. In a separate but related context, the same model is also believed to capture some essential aspects of information loss and recovery in black holes, even without a precise microscopic model of black holes as may emerge from an eventual theory of quantum gravity~\cite{HaydenPreskill}. Here, one thinks of the thermal reservoir $\mathcal{H}_E$ as representing the information already contained within a black hole, which is in the thermal state $\hat{\rho}_{\beta E}$ due to entanglement with the Hawking radiation emitted before $t=0$ (a separate quantum system $\mathcal{H}_R$ that does not interact with $\mathcal{H}_E$). The system $\mathcal{H}_S$ then represents infalling matter into the black hole, and the above setup models how infalling information spreads among a black hole's degrees of freedom, and whether it can be recovered at any given time from the previously emitted Hawking radiation~\cite{HaydenPreskill}.

In both cases, a problem of fundamental interest is to determine the fastest allowed scrambling time $t_s$, after which $\mathcal{H}_S$ retains no memory of its initial state up to thermodynamic parameters like temperature. From the perspective of quantum statistical mechanics, this amounts to constraining the earliest time after which equilibrium statistical mechanics can apply to any subsystem of interest. Separately, a scrambling time of $t_s \gtrsim O(\beta \ln N)$ has been 
conjectured in sufficiently $k$-local systems of $N=N_S+N_E$ particles~\cite{HaydenPreskill, SekinoSusskind}, based on arguments attempting to reconcile models of quantum information in black holes with ideas in quantum gravity. A number of studies of the scrambling time with different notions of scrambling appear to be consistent with this conjecture~\cite{LashkariFastScrambling, MSSotocBound, MaldacenaStanford, BentsenGuLucasScrambling} in certain classes of systems. \edit{However, it has not been clear how to even formulate the problem of fast scrambling to enable precise results, and indeed, exceptions to these bounds are also known with regard to qualitative formulations for ``typical'' quantum systems~\cite{MSSotocBound}. Further,} it is not expected that a nontrivial bound on $t_s$ exists for a general quantum mechanical system without any restrictions on the nature of interactions~\cite{LashkariFastScrambling, ChaosComplexityRMT}.

\textit{Qualitative summary}--- In this Letter, we tackle the challenge of \edit{formulating and} proving a bound on fast scrambling assuming ``just the validity of quantum mechanics'' (as posed in Ref.~\cite{LashkariFastScrambling}), going beyond the specific models of local interactions that were previously considered necessary. Our approach is enabled by direct connections between quantum dynamics and the structure of the energy spectrum developed in Refs.~\cite{dynamicalqergodicity, dynamicalqspeedlimit}, assuming little beyond the existence of quantum energy levels. In particular, due to a refinement of the energy-time uncertainty principle resulting from these connections~\cite{dynamicalqergodicity, dynamicalqspeedlimit}, the problem of formulating speed limits on \textit{infinite temperature} scrambling within quantum mechanics can be mapped to the asymptotic properties of the Fourier transform of the density of states (distribution of energy levels) of a Hamiltonian system~\cite{dynamicalqspeedlimit}. In this work, we extend this result in two ways: (1)~We generalize the previously obtained speed limit to a reservoir in an arbitrary initial state $\hat{\rho}_{\beta E}$, including finite temperature states [Eqs.~\eqref{eq:dynamicalinequalityagain}, \eqref{eq:regSFFdef}]; (2)~We identify (and derive a quantitative version of) a rigorous theorem on the asymptotics of Fourier transforms that provides an exact speed limit for any Hamiltonian system.
Surprisingly, our speed limit holds for arbitrarily nonlocal systems, including random matrix Hamiltonians~\cite{Mehta} \editb{as demonstrated in Fig.~\ref{fig:bound}}. It establishes a quantitative limit on fast scrambling [Eq.~\eqref{eq:ts_scr_bound_main}] that is logarithmic in the scrambled entanglement entropy [Eq.~\eqref{eq:ts_entropy_form}], as a property of quantum mechanics itself rather than any particular class of systems. \edit{Overall, our result constitutes the first rigorous formulation, as well as proof, of the finite temperature fast scrambling problem in any quantum mechanical system, including $k$-local and nonlocal systems. In addition, our formulation applies rigorously to finite-dimensional systems, such as a collection of a finite number of interacting qubits, which may allow it to be quantitatively tested in available experimental platforms.}

\begin{figure}
    \centering
    \includegraphics[width=0.9\columnwidth]{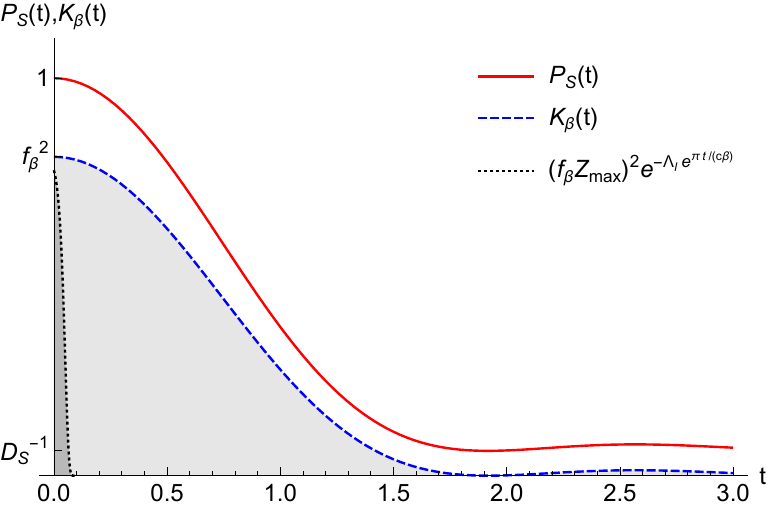}
    \caption{\editb{An illustration of our finite temperature speed limit on scrambling, as measured by the mean return probability (solid red) $P_S(t)$, from the regularized spectral form factor (dashed blue) $K_{\beta}(t)$ [Eq.~\eqref{eq:dynamicalinequalityagain}], and the universal double-exponential bound (dotted black) on the latter [combining Eqs.~\eqref{eq:Kbeta_tildos_relation} and \eqref{eq:doubleexpbound1}], for a system with $N=10$ total qubits split into a subsystem with $N_S = 4$ and a reservoir with $N_E = 6$. The Hamiltonian $H$ is taken to be a \textit{single} realization of the Gaussian Unitary Ensemble~\cite{Haake, Mehta} of random Hermitian matrices in a product basis~\cite{NielsenChuang} of $S$ and $E$, normalized to have an ensemble-averaged density of states spanning $E_n \in [-2,2]$. The reservoir is initialized to $\rho_{\beta E} \propto \exp[-\beta \Tr_S(H)]$ with $\beta = 10$, while $S$ is initialized to fixed basis states in the product basis. The mathematical parameters in the double exponential bound [Eqs.~\eqref{eq:doubleexpbound1} and \eqref{eq:exprateresult}] are manually chosen to be $\tau_1 = -0.05$, $\tau_2 = 0$ (corresponding to $c\beta = 0.05$; see also Fig.~\ref{fig:PartitionFunc}), and $\ell = 1$ (exceptional set length $c\beta \ell/\pi \approx 0.016$). Random matrices are far from saturating the double exponential bound~\cite{dynamicalqspeedlimit}, but provide a convenient illustration of both bounds (see the supplement~\cite{supplementfinitetemp} for a formal example of an admissible $K_{\beta}(t)$ for which a system can nearly saturate the bound).
    }}
    \label{fig:bound}
\end{figure}

\textit{Quantitative overview of results}--- To formulate our results in a manner similar to Ref.~\cite{dynamicalqspeedlimit}, we work with the mean return probability of the computational basis states in $\mathcal{H}_S$ (i.e., selecting $\proj_k$ with $\alpha = k$ in Eq.~\eqref{eq:projdef}) under evolution by a Hamiltonian $\hat{H}$, given the corresponding initial state $\hat{\rho}_k(0)$ in Eq.~\eqref{eq:initstatedef}:
\begin{equation}
    P_S(t) \equiv \frac{1}{D_S}\sum_{k=1}^{D_S} \Tr[e^{-i\hat{H}t} \hat{\rho}_k(0) e^{i\hat{H}t} \proj_k].
    \label{eq:newscrpsdef}
\end{equation}
This quantity partially measures the extent to which the \textit{computational basis}, among all the bases in $\mathcal{H}_S$, retains information about the initial state, and constrains the generation of entanglement between the two systems which is at the root of quantum thermalization~\cite{tumulka_CT, CanonicalTypicalityPSW, NormalTypicality, subETH}. For example, the average purity of the reduced states $\hat{\rho}_{k,S}(t) \equiv \Tr_E\hat{\rho}_{k}(t)$ is constrained by~\cite{supplementfinitetemp},
\begin{equation}
    \frac{1}{D_S} \sum_{k=1}^{D_S} \Tr_S\left[\hat{\rho}_{k,S}^2(t)\right] \geq [P_S(t)]^2.
    \label{eq:purity_and_ps}
\end{equation}
Initially, $P_S(0) = 1$, signifying complete retention of this information and pure basis states in $\mathcal{H}_S$. At subsequent times, we expect the scrambling process to cause $P_S(t)$ to decay to a ``scrambled'' value $P_{\scr}$. For example, when information about the initial state is completely and uniformly delocalized (to leading order) over the $D_S$ basis states $\lvert k\rangle_S$, we have
$P_{\scr} \sim D_S^{-1}$,
which is typically necessary for maximal entanglement~\cite{dynamicalqspeedlimit}.
Here, we use formal asymptotic notation~\cite{knuth1976asymptotic}, with $\omega, \Omega, \Theta, O, o$ representing $>, \geq, =, \leq, <$ up to any multiplicative constant, and $\gtrsim$, $\sim$, $\lesssim$ for leading-order (in)equalities, when a parameter such as $D$ diverges.
To maintain generality, we do not assume a specific scrambled value at this stage, and we allow $D$ to remain finite without an explicit limit, except where asymptotic notation is used.

Our main concern is the \textit{sustained}~\cite{dynamicalqspeedlimit} scrambling time $t_s$, \edit{defined via Eq.~\eqref{eq:scrtimecriterion2} below}, after which $P_S(t)$ remains no larger than the corresponding scrambled value $P_{\scr}$. We must also implicitly only consider times $\lvert t\rvert < T_D$, where $T_D$ is a cutoff beyond all time scales of physical interest, chosen to avoid quantum recurrences~\cite{QuantumRecurrences} to large $P_S(t)$. Thus, $t_s$ is set by the condition
\begin{equation}
    P_S(t > t_s) \leq P_{\scr}.
    \label{eq:scrtimecriterion2}
\end{equation}
Given this setup, we can summarize our main result. We find that $t_s$ is bounded in any quantum mechanical system in terms of the following three quantities: (1)~the scrambled return probability $P_{\scr}$, (2)~the information-theoretic \textit{fidelity}~\cite{NielsenChuang} $0 \leq f_{\beta} \leq 1$ between the thermal state $\hat{\rho}_{\beta E}$ and the infinite temperature maximally mixed state $\hat{\rho}_{0 E} \equiv \idop_E/D_E$ in $\mathcal{H}_E$,
\begin{equation}
    f_{\beta} \equiv \Tr_E\sqrt{\sqrt{\hat{\rho}_{0 E}}\ \hat{\rho}_{\beta E}^{\vphantom{1/1}} \sqrt{\hat{\rho}_{0 E}}} = \frac{1}{\sqrt{D_E}} \Tr_E\left[\hat{\rho}_{\beta E}^{1/2}\right],
    \label{eq:fidelityinfdef}
\end{equation}
\edit{which measures how ``mixed'' the state $\rho_{\beta E}$ is in $\mathcal{H}_E$}, and (3)~an effective partition function $Z_{\beta}(\tau) = \tildoseff(-i\tau)$ [see Eq.~\eqref{eq:tildoseffdef}] in a Euclidean time interval containing $\tau=0$:
\begin{equation}
    \tau \in [\tau_1,\tau_2],\ \tau_1 \leq 0 \leq \tau_2, \text{ of width } (\tau_2-\tau_1) \equiv c\beta.
    \label{eq:horizontalstrip}
\end{equation}
\edit{This effective partition function encodes the thermodynamic properties of the distribution $\hat{\rho}_{\beta E}^{1/2}$ that appears in the expression in Eq.~\eqref{eq:fidelityinfdef} for the fidelity of the thermal state, if the system interacts with an additional (fictitious) external bath at effective temperature $\tau^{-1}$.}

Specifically, we show that for \textit{every} interval $[\tau_1,\tau_2]$ as in Eq~\eqref{eq:horizontalstrip}, the sustained scrambling time satisfies:
\begin{equation}
    t_s \geq \frac{c \beta}{\pi} \ln \left[\frac{1}{\expratemain_{\text{eff}}}\ln \left(\frac{f_{\beta}^2 Z_{\max}^2}{P_{\scr}}\right)\right].
    \label{eq:ts_scr_bound_main}
\end{equation}
The parameters $Z_{\max} = \max_{\tau \in [\tau_1,\tau_2]} Z_{\beta}(\tau)$ (attained~\cite{supplementfinitetemp} at $\tau = \tau_1$ or $\tau_2$) and $\expratemain_{\text{eff}}$ [Eqs.~\eqref{eq:exprateresult}, \eqref{eq:ts_universal_bound}] depend on $Z_{\beta}(\tau)$ within the interval. \edit{Strictly speaking, Eq.~\eqref{eq:ts_scr_bound_main} constitutes a family of bounds, one for each choice of the Euclidean time interval.} One can further optimize the choice of interval in Eq.~\eqref{eq:horizontalstrip}, including $c\beta$, to yield the tightest bound on $t_s$, but this depends on the behavior of $Z_{\beta}(\tau)$ and must be done in system-specific ways. Eq.~\eqref{eq:ts_scr_bound_main} provides a direct, exact bound on the generation of entanglement [via Eq.~\eqref{eq:purity_and_ps}], and applies to arbitrary Hamiltonians $\hat{H}$, in contrast to the (conjectured) bound on chaos in out-of-time-ordered correlators (OTOCs) of local operators that is shown to hold only for special classes of Hamiltonians~\cite{MSSotocBound, BentsenGuLucasScrambling, LucasEntanglementVsOTOC, LucasReview}. Further, the two bounds do not necessarily imply each other and refer to distinct physical processes associated with scrambling~\cite{BentsenGuLucasScrambling, LucasEntanglementVsOTOC, dynamicalqspeedlimit, TezukaSYKScrambling, harrow2021separation}.

We can alternatively directly constrain $t_s$ in terms of a minimum 2nd R\'{e}nyi entanglement entropy~\cite{HorodeckiEntanglementReview} $S_{2,S}$ required for a state to be considered scrambled in $\mathcal{H}_S$:
\begin{equation}
S_2[\hat{\rho}_{k,S}^2(t>t_s)] \equiv -\ln\Tr[\hat{\rho}_{k,S}^2(t)] \geq S_{2,S}.
\label{eq:entropyscrambling}
\end{equation}
\edit{While this is an alternative to Eq.~\eqref{eq:scrtimecriterion2} that is perhaps more directly relevant in studies of many-body entanglement generation, $S_{2,S}$ also implies a specific choice of $P_{\scr}$. In particular, by} Eq.~\eqref{eq:purity_and_ps}, $P_{\scr}$ can
be replaced~\cite{supplementfinitetemp} by $\exp(-S_{2,S}/2)$ in Eqs.~\eqref{eq:scrtimecriterion2} and \eqref{eq:ts_scr_bound_main}. 
To simplify the resulting bound on $t_s$, we consider a thermodynamic limit $N_S, N_E\to\infty$, such that $Z_{\beta}(\tau)$ remains $O(1)$ in $[\tau_1,\tau_2]$ and $f_{\beta} \geq \exp[-o(S_{2,S})]$ is not too small \edit{--- essentially, a regime in which the leading contribution is from the final scrambled entropy $S_{2,S}$ rather than the properties of the initial state}. In this regime, we obtain~\cite{supplementfinitetemp} to leading order:
\begin{equation}
    t_s \gtrsim \frac{c\beta}{\pi} \ln S_{2,S}.
    \label{eq:ts_entropy_form}
\end{equation}
Further, a canonical thermal state typically admits $c = {(1-\epsilon)/2}$ with a small $\epsilon > 0$ for $Z_{\beta}(\tau) = O(1)$ [as motivated near Eq.~\eqref{eq:partitionfuncapprox}]. For volume-law (including maximal) entanglement~\cite{Nandkishore, MBLThermalization}, $S_{2,S}\sim \mu N_S$ for $N_S < N_E$. Consequently, the conjectured~\cite{HaydenPreskill, SekinoSusskind} fast scrambling bound $t_s = \Omega(\beta \ln N)$ holds for subsystems with $N_S \sim N^\kappa$ (where $0 < \kappa \lesssim 1$).

\textit{Effective nonunitary dynamics at finite temperatures}--- We begin our technical derivation with a brief summary of a result of Ref.~\cite{dynamicalqspeedlimit}. There, we derived a universal speed limit on $P_S(t)$ under arbitrary time-dependent quantum operations~\cite{NielsenChuang} $\hat{\rho}_k(t) = \sum_{r=1}^M \kraus{r}{t} \hat{\rho}_{k,0} \krausd{r}{t}$, but only for a special choice of initial states $\hat{\rho}_{k,0} = \proj_k/D_E$ corresponding to the infinite temperature state $\hat{\rho}_{0 E}$ in Eq.~\eqref{eq:initstatedef}.
To adapt this bound for the more general family of initial states in Eqs.~\eqref{eq:initstatedef} and \eqref{eq:newscrpsdef}, we can apply a simple trick that trades off unitary dynamics with an arbitrary initial state for nonunitary dynamics with infinite temperature initial states.

Specifically, in Eq.~\eqref{eq:newscrpsdef}, we write
\begin{align}
    \hat{\rho}_k(t) = e^{-i\hat{H}t} \hat{\rho}_k(0) e^{i\hat{H}t} &= \frac{1}{D_E}\kraus{1}{t}\proj_k\krausd{1}{t}, \label{eq:timeevolvingstateKraus} \\
    \text{with } \kraus{1}{t} &\equiv e^{-i\hat{H}t} \hat{\rho}_{\beta E}^{1/2} D_E^{1/2}. \nonumber
\end{align}
Thus, $P_S(t)$ can also be regarded as the mean return probability of the infinite temperature initial states $\hat{\rho}_{k,0} = \proj_k/D_E$ with effective time-dependent nonunitary dynamics given by $\kraus{1}{t}$. In this case, the speed limit of Ref.~\cite{dynamicalqspeedlimit} \edit{can be expressed as: 
\begin{equation}
\frac{1}{D}\sum_k \Tr[\kraus{1}{t} \proj_k \krausd{1}{t}\proj_k] \geq \frac{1}{D^2} \left\lvert\Tr[\kraus{1}{t}]\right\rvert^2.
\label{eq:dynamicalinequalityprimitive}
\end{equation}
This bound follows from applying the Cauchy-Schwarz and triangle inequalities to the left hand side, and can be understood as a generalization of the energy-time uncertainty principle to quantum operations, with unique sensitivity to the long timescales associated with many-body processes. In our present context, Eq.~\eqref{eq:dynamicalinequalityprimitive}} gives the following bound~\cite{supplementfinitetemp}, which underlies the main results of this Letter \editb{(see Fig.~\ref{fig:bound})}:
\begin{equation}
    P_S(t) \geq K_{\beta}(t). \label{eq:dynamicalinequalityagain}
\end{equation}
In the above inequality,
\begin{equation}
    K_{\beta}(t) = \frac{D_E}{D^2}\left\lvert\Tr\left(e^{-i\hat{H}t} \hat{\rho}_{\beta E}^{1/2}\right)\right\rvert^2 \label{eq:regSFFdef}
\end{equation}
is the \textit{regularized} spectral form factor (SFF) corresponding to the density operator $\hat{\rho}_{\beta E}$, occurring here as a special case of the generalized SFF~\cite{dynamicalqspeedlimit} for time-dependent quantum operations $K(t) \equiv D^{-2}\sum_{r=1}^M \left\lvert \Tr \kraus{r}{t}\right\rvert^2$, which extends and unifies time-dependent~\cite{timeDependentSFF} and dissipative~\cite{dissipativeFormFactor0, dissipativeFormFactor1, dissipativeFormFactor2, dissipativeFormFactor3} SFFs.

For context, we note that SFFs are a key staple of quantum chaos that characterize energy level correlations~\cite{Haake} and the presence of observable-independent ergodic properties in quantum dynamics~\cite{dynamicalqergodicity, bakeranomalieserg}. Further, they can be directly measured in experiments~\cite{SFFmeas, pSFF, pSFFexpt1}. Eq.~\eqref{eq:regSFFdef} specifically belongs to a family of regularized SFFs~\cite{BlackHoleRandomMatrix, SonnerThermodouble, ExpRamp4, SFFinflectionbound, matsoukas2023quantum, matsoukas2023unitarity, KurchanRegSFFzeros} and correlators~\cite{MSSotocBound, MaldacenaStanford}, where the splitting of roots of a density operator between different factors (usually for analytical convenience) has been regarded as somewhat artificial to engineer in experiments~\cite{SFFmeas, OTOC2022ExptFiniteTemp}. That the regularized SFF $K_{\beta}(t)$ universally constrains pure Hamiltonian dynamics via Eq.~\eqref{eq:dynamicalinequalityagain} now provides an exact operational physical justification for regularization, at least in one context. 

\textit{Properties of regularized density of states}--- From Eqs.~\eqref{eq:dynamicalinequalityagain}, \eqref{eq:regSFFdef}, it is clear that we can obtain universal bounds on the decay of $P_S(t)$ if we can constrain the decay of $K_{\beta}(t)$ for any Hamiltonian system. To do so, it is useful to focus on its ``square root'' up to constant factors, the trace of $\kraus{1}{t}$,
\begin{equation}
    \tildoseff(t) \equiv \frac{\Tr \kraus{1}{t}}{\Tr \kraus{1}{0}} = \frac{D_E^{1/2}}{f_{\beta}D}\sum_{n} \langle E_n\rvert\hat{\rho}_{\beta E}^{1/2}\lvert E_n\rangle e^{-iE_n t},
    \label{eq:tildoseffdef}
\end{equation}
normalized to $\tildoseff(0) = 1$, so that $\lvert \tildoseff(t \geq 0)\rvert \leq 1$. $\tildoseff(t)$ is the Fourier transform of the \textit{regularized} density of states (a probability distribution that integrates to $\tildoseff(0) = 1$),
\begin{equation}
\doseff(E) = \frac{D_E^{1/2}}{f_{\beta}D}\sum_n \langle E_n\rvert\hat{\rho}_{\beta E}^{1/2}\lvert E_n\rangle \delta(E-E_n) \geq 0.
\end{equation}
In terms of this quantity, we have
\begin{equation}
    K_{\beta}(t) = f_{\beta}^2\left\lvert \tildoseff(t)\right\rvert^2.
    \label{eq:Kbeta_tildos_relation}
\end{equation}
Now, we can use this in conjunction with Eqs.~\eqref{eq:scrtimecriterion2} and \eqref{eq:dynamicalinequalityagain} to directly relate the scrambling time $t_s$ to the regularized density of states,
\begin{equation}
    \left\lvert \tildoseff(t > t_s)\right\rvert^2 \leq \frac{P_{\scr}}{f_{\beta}^2},
    \label{eq:dosscramblingcriterion_gen}
\end{equation}
generalizing the corresponding infinite temperature result of Ref.~\cite{dynamicalqspeedlimit}. For this to give a nontrivial $t_s > 0$, we must have $f_{\beta}^2 > P_{\scr}$ (as $\lvert \tildoseff(t)\rvert \leq 1$). This condition can always be satisfied for any given $P_{\scr}$ by any $\hat{\rho}_{\beta E}$ with sufficient entanglement entropy in $\mathcal{H}_E$~\cite{supplementfinitetemp}, and for any given initial state $\hat{\rho}_{\beta E}$ by choosing $N_S$ large enough to have sufficiently small $P_{\scr} < f_{\beta}^2$ [e.g., if $P_{\scr} \sim D_S^{-1}$].

To constrain the decay of $\tildoseff(t)$, we extend the time variable to the complex plane, $\tcomp = t - i \tau$, with $t, \tau \in \mathbb{R}$. In this case, $\tildoseff(\tcomp)$ is a weighted sum of $D$ (analytic) exponentials $e^{-iE_n \tcomp}$, and is therefore analytic everywhere on $\mathbb{C}$. It is also real-valued on the imaginary axis, with $\tildoseff(-i\tau)$ = $Z_{\beta}(\tau)$ being the partition function of the regularized density of states at inverse temperature (or Euclidean time) $\tau$. From Eq.~\eqref{eq:tildoseffdef}, this partition function sets an upper bound on $\lvert \tildoseff(\tcomp)\rvert$ at each $\tau$:
\begin{equation}
    \lvert \tildoseff(t-i\tau)\rvert \leq Z_{\beta}(\tau),\ \forall\ t \in \mathbb{R},
    \label{eq:ridgeproperty}
\end{equation}
which is saturated (at least) at $t = 0$.

While $Z_{\beta}(0) = 1$, it can diverge for other values of $\tau$, making $\tildoseff(\tcomp)$ unbounded in a thermodynamic limit. For example, for a canonical thermal state $\hat{\rho}_{\beta E}$ with $N_S \ll N_E$ whose energy eigenbasis matrix elements behave as $e^{-\beta E_n}$,
\begin{equation}
    Z_{\beta}(\tau) \approx \frac{\sum_n e^{-(\beta+2\tau)E_n/2}}{\sum_n e^{-\beta E_n/2}}.
    \label{eq:partitionfuncapprox}
\end{equation}
\edit{The precise behavior of this partition function depends on the density of states $\dos(E)$ of the full system. We note that the \textit{intrinsic} thermodynamic inverse temperature of the system at energy $E$ (to which smaller external ``probe'' systems, interacting weakly with our system at energy $E$, may be expected to thermalize) is given by $\beta_{\text{int}}(E) = \diff \ln \dos(E)/\diff E$, which is separate from the inverse temperature $\beta$ set by the initial state $\rho_{\beta E}$ (in practice, determined by the additional reservoir $\mathcal{H}_R$ initially coupled to $\mathcal{H}_E$)~\cite{LLStatMech, ReichlStatMech}. For a representative argument, we focus on systems in which one expects the intrinsic temperature $1/\beta_{\text{int}}(E)$ to increase with energy, as is frequently the case in thermodynamics, which corresponds to $\dos(E)$ growing slower than exponentially with $E$. Here,} for most $-\beta/2 < \tau \leq 0$, the numerator remains exponentially suppressed, and we expect $Z_{\beta}(\tau) = \Theta(1)$ if the density of states doesn't grow exponentially with $E_n$. However, for $\tau \approx -\beta/2$, the exponential suppression in the numerator is lost and $Z_{\beta}(-\beta/2) \approx D_E^{1/2}/f_{\beta}$ diverges as $N_E \to \infty$. This motivates considering horizontal strips containing $\tau=0$, restricted to a specific width $c\beta$ as in Eq.~\eqref{eq:horizontalstrip}, in which quantities remain bounded. While this argument suggests that a strip in $(-\beta/2,0]$ with $c = (1-\epsilon)/2$ is appropriate for thermal states \edit{in many physical systems relevant for thermodynamics}, we formulate our results for an arbitrary width $c\beta \geq 0$ to retain exactness and universal applicability to arbitrary \edit{systems}. \editb{The choice of strip for the illustration in Fig.~\ref{fig:bound} is depicted in Fig.~\ref{fig:PartitionFunc}}.

\begin{figure}
    \centering
    \includegraphics[width=0.9\columnwidth]{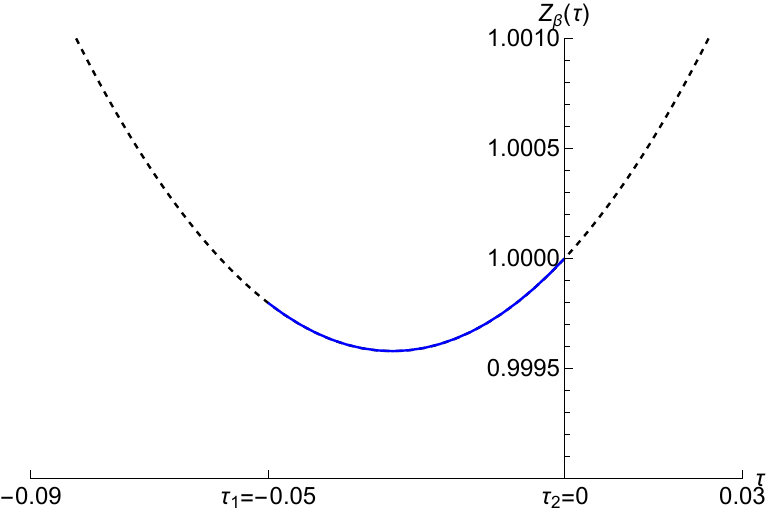}
    \caption{\editb{The partition function $Z_{\beta}(\tau) = N_{\beta}(-i\tau)$ [see Eq.~\eqref{eq:tildoseffdef}] illustrated for the finite-dimensional random matrix system with $D=1024$ in Fig.~\ref{fig:bound}, with the portion within the chosen interval $[\tau_1,\tau_2] = [-0.05,0]$ (i.e., $c\beta = 0.05$) highlighted in solid blue, and the rest of the curve in dashed black. This illustrates, for example, the unbounded nature of the partition function outside a $\tau$-interval of choice, and that smaller values of $Z_{\beta}(\tau)$ may occur for $\tau < 0$.}}
    \label{fig:PartitionFunc}
\end{figure}

Here, we also note previous physical bounds involving similar properties of analyticity and boundedness in a strip, including that expected for regularized OTOCs~\cite{MSSotocBound}, as well as at a point of inflection of the regularized SFF~\cite{SFFinflectionbound}, which however differ qualitatively and quantitatively from the results of interest [e.g., Eq.~\eqref{eq:doubleexpbound1}] in this Letter. \edit{Both previous bounds derive from a specific mathematical result, however the bound on OTOCs~\cite{MSSotocBound} is not rigorous~\cite{LucasReview} and has known exceptions, while the bound on regularized SFFs~\cite{SFFinflectionbound} is an instantaneous speed limit such as those in Ref.~\cite{SonnerThermodouble}, without sufficient senstivity to determine late-time dynamics beyond the appropriate $O(1)$ timescales. Such sensitivity is, however, crucial to constrain the scrambling (via entanglement generation) of large subsystems, whose timescales can scale with $N_S$. In addition, these bounds qualitatively assume a thermodynamic limit $N\to\infty$ in which the relevant analyticity properties hold; this can prevent a rigorous application of such bounds in experiments with a finite number of qubits $N$. We show below that a different quantitative mathematical result on the asymtotics of analytic functions can be extended to address both issues rigorously for our purposes, even for finite-dimensional systems [represented by the finite sum in Eq.~\eqref{eq:tildoseffdef}], without requiring any qualitative arguments about the thermodynamic limit.}

\textit{Universal bound on fast scrambling}--- In our case, given that $\tildoseff(\tcomp)$ is analytic and bounded in the strip given by Eq.~\eqref{eq:horizontalstrip}, results from the theory of analytic functions~\cite{AhlforsHeins, Hayman, LinnikOstrovskii, Hayman_v2, Koosis} constrain $\tildos_{\beta}(t\geq 0)$ to decay no faster than a double exponential in $t$ for \textit{nearly all} $t>0$, \textit{except} in an ``exceptional set'' $\texcept \subset \mathbb{R}_{>0}$ of finite length in $t$. In particular, recalling that $Z_{\max} = \max_{\tau \in [\tau_1,\tau_2]} Z_{\beta}(\tau)$, we have~\cite{supplementfinitetemp} \editb{(see Fig.~\ref{fig:bound} for an example):}
\begin{equation}
    \lvert \tildoseff(t \notin \texcept)\rvert^2 \geq Z_{\max}^2 \exp\left[-\expratemain \exp\left(\frac{\pi t}{c\beta}\right)\right],
    \label{eq:doubleexpbound1}
\end{equation}
where $\expratemain$ is a finite constant~\cite{AhlforsHeins, Hayman}.
For intuition behind the exceptional set $\texcept$, consider~\cite{Reimann2016, dynamicalqergodicity, dynamicalqspeedlimit} the function $\tildoseff(t) = \sinc^2(\pi t/\beta)$, which is oscillating and vanishes at the nodes $t = n\beta$ for positive integers $n$. While its asymptotic behavior is generally $t^{-2}$, satisfying Eq.~\eqref{eq:doubleexpbound1} for most $t$, the oscillations bring it below the bound of Eq.~\eqref{eq:doubleexpbound1} in a (shrinking with $n$) neighborhood of each node, which comprise the finite-length set $\texcept$. The presence of the exceptional set excludes all discrete-time systems from a nontrivial application of Eq.~\eqref{eq:doubleexpbound1}, such as Haar random unitaries that maximally scramble $\mathcal{H}_S$ after just a single time-step $t=1$, but \textit{not} their continuous-time Hamiltonian extrapolations.

However, in this form, Eq.~\eqref{eq:doubleexpbound1} is not sufficiently predictive for our purposes. This is because the ``finite'' constant $\expratemain$ and the ``finite'' duration of $\texcept$ are undetermined, and could still be larger or smaller than any scale of interest for a system with finite $D$.
Fortunately, we can adapt the methods of Ref.~\cite{Hayman_v2} to explicitly obtain their values~\cite{supplementfinitetemp}. We find that a rescaled length of $\texcept$ can be made smaller than any chosen $\ell > 0$,
\begin{equation}
    \int_{\texcept}\diff t \leq \frac{c\beta}{\pi}\ell
    \label{eq:exceptionalsetlength},
\end{equation}
if one chooses $\expratemain$ according to 
\begin{equation}
   \expratemain_{\ell} = \frac{(2+\ell)}{\ell}\left\lbrace\min_{\tau \in [\tau_1,\tau_2]}\ \frac{2(\pi^2+8)\ln\left[\frac{Z_{\max}}{Z_{\beta}(\tau)}\right]}{\cos\left[\frac{\pi}{c\beta} \left(\tau-\frac{\tau_1+\tau_2}{2}\right)\right]}\right\rbrace.
   \label{eq:exprateresult}
\end{equation}

We can now obtain the scrambling time from Eq.~\eqref{eq:dosscramblingcriterion_gen}, with the caveat that the exceptional set introduces an uncertainty in the scrambling time with range $[t_s, t_s+ (c\beta\ell/\pi)]$ from Eq.~\eqref{eq:exceptionalsetlength}. We get~\cite{supplementfinitetemp}, using Eq.~\eqref{eq:dosscramblingcriterion_gen}, \eqref{eq:doubleexpbound1} and \eqref{eq:exprateresult},
\begin{equation}
    t_s \geq \frac{c\beta}{\pi} \left\lbrace \ln \left[\frac{1}{\expratemain_{\ell}}\ln \frac{f_{\beta}^2 Z_{\max}^2}{P_{\scr}}\right]-\ell\right\rbrace,
    \label{eq:ts_universal_bound}
\end{equation}
where the second term accounts for the uncertainty by subtracting the length of $\texcept$ from the bound obtained. 
To optimize tightness, we choose $\ell$ to maximize the right hand side of Eq.~\eqref{eq:ts_universal_bound}. We get $\ell = \sqrt{3}-1$, from which we recover our main universal result, Eq.~\eqref{eq:ts_scr_bound_main}, with  $\expratemain_{\text{eff}} \equiv e^{\sqrt{3}-1}\expratemain_{(\sqrt{3}-1)}$.

Further, an important question is whether a near-optimally decaying double exponential function as in Eq.~\eqref{eq:doubleexpbound1} can emerge in a system with a nonnegative (regularized) density of states $\doseff(E) \geq 0$ (e.g., \cite{dynamicalqspeedlimit}). This is indeed the case: a class of examples was constructed in Refs.~\cite{Ostrovskii}, \cite[p.35]{LinnikOstrovskii}, 
and we construct other (simpler) examples~\cite{supplementfinitetemp} by
taking the thermodynamic limit of $\tildoseff(t)$ to be the 
convolution ${g*g}$, where $g$ is  the double-exponentially decaying function $g(t)=\exp(-e^{\frac{\pi}{\cbeta}t})\exp(-e^{-\frac{\pi}{\cbeta}t})$.
Given sufficiently (including Haar~\cite{Haake}) random energy eigenstates, these examples scramble~\cite{supplementfinitetemp} by $t_s \lesssim (2c\beta/\pi)\ln N_S$, nearly saturating Eq.~\eqref{eq:ts_entropy_form}~\cite{dynamicalqspeedlimit, Reimann2016, ChaosComplexityRMT, CotlerHunterJones2}.

\textit{Discussion}---We consider the implications of our results in light of the fast scrambling conjecture~\cite{SekinoSusskind}, which states that (a) $t_s \gtrsim O(\beta \log N)$ for some definition of scrambling time and some class of systems (neither rigorously specified~\cite{LashkariFastScrambling, BentsenGuLucasScrambling, LucasEntanglementVsOTOC, LucasReview}), and (b)  quantum systems believed to describe black holes (again not rigorously specified) saturate this bound. For a successful analysis, we must specify at least one of the three unspecified features and examine its implications for the others.

In this Letter, we have considered the notion of scrambling in Eq.~\eqref{eq:scrtimecriterion2} that directly constrains the time required for sustained entanglement generation and the applicability of equilibrium statistical mechanics in subsystems. In this case, Eq.~\eqref{eq:ts_scr_bound_main} provides a universal quantum mechanical bound on information scrambling by \textit{any} Hamiltonian system; Eq.~\eqref{eq:ts_entropy_form} further shows that even the conjectured $t_s = \Omega(\beta \ln N)$ form holds \textit{universally} for volume-law entanglement in a certain thermodynamic regime. This completely specifies and establishes statement (a) of the fast scrambling conjecture for this notion of scrambling, generalizing it beyond any specific assumptions of interaction structure.

For statement (b), we first consider the example of the ``maximally chaotic'' Sachdev-Ye-Kitaev models~\cite{MaldacenaStanford}, which saturate the OTOC bound for local correlators~\cite{MSSotocBound}. A key feature of these models is a zero-temperature entropy (extensive clustering of states near the ground state)~\cite{MaldacenaStanford, SachdevEntropy, GKSTcomplexSYK} that is believed to directly capture some aspects of black hole physics~\cite{SachdevEntropyHolography, SachdevEntropyBlackHoles}. However, this feature leads to a subleading slow decay~\cite{BlackHoleRandomMatrix} of $K_{\beta}(t)$, giving $t_s$ exponential in $N_S$ (for large $N_S$) as shown in Ref.~\cite{dynamicalqspeedlimit} (see also Ref.~\cite{KurchanSYKentropyglass} for an interesting parallel to glassy metastable states). Thus, for scrambling via entanglement generation, we can limit the applicability of statement (b) by noting that systems with a nonvanishing zero temperature entropy are not fast scramblers in large subsystems. \editb{This includes standard random matrix ensembles~\cite{Haake, Mehta}, as illustrated in Fig.~\ref{fig:bound}.} At the same time, the discussion after Eq.~\eqref{eq:ts_universal_bound} illustrates what the energy spectrum of a ``nearly fast scrambler'' may look like, showing that systems nearly saturating Eq.~\eqref{eq:ts_entropy_form} do \textit{formally} exist. An intriguing open challenge is to address the following question: What is the natural physical setting in which we can expect to find such fast scramblers in the sense of sustained entanglement generation?

\vspace{2em}

\begin{acknowledgments}
\textit{Acknowledgments}--- A.V. acknowledges useful discussions with Jorge Kurchan during the program -- Stability of Quantum Matter in and out of Equilibrium at Various Scales (code: ICTS/SQMVS2024/01) at the International Center for Theoretical Sciences, Bengaluru, India.
This work was supported by the U.S. Department of Energy, Office of Science, Basic Energy Sciences under Award No. DE-SC0001911.
The authors acknowledge the University of Maryland supercomputing resources (https://hpcc.umd.edu) made available for generating the plots in Section \ref{sec:nearlyfastscramblers} of the supplement
to this paper. 
\end{acknowledgments}

\vspace{-1.6em}
\bibliography{FiniteTempBibliography}

\clearpage
\appendix

\setcounter{secnumdepth}{4}   
\setcounter{section}{0}

\onecolumngrid

\begin{center}
\textbf{\large Proof of a Universal Speed Limit on Fast Scrambling in Quantum Systems}\vspace{0.2em}

\textbf{\large Supplemental Material}

\vspace{1em}

{\normalsize Amit Vikram,$^{1}$\ Laura Shou,$^{1,2}$\ and Victor Galitski$^{1}$}\vspace{0.2em}

$^{1}$\textit{\small Joint Quantum Institute and Department of Physics, University of Maryland, College Park, MD 20742, USA}

$^{2}$\textit{\small Condensed Matter Theory Center, Department of Physics, \\ University of Maryland, College Park, MD 20742, USA}

\vspace{1em}
\end{center}

\maketitle

In this supplement, we derive various results presented in the main text. In Sec.~\ref{sec:scramblingandentanglement}, we derive quantitative relations between scrambling in terms of return probabilities and in terms of entanglement entropies [supporting Eqs.~\maintexteqpurityandps{} and \maintexteqtsentropyform{} of the main text]. In Sec.~\ref{sec:psktgeneralization}, we derive the inequality $P_S(t) \geq K_{\beta}(t)$ that formulates the energy-time uncertainty principle with sensitivity to microscopic values of $P_S(t)$ and nontriviality at arbitrarily long times (as in Ref.~\cite{dynamicalqspeedlimit}), generalized to any initial state of the external system including finite temperature states [Eq.~\maintexteqdynamicalinequalityagain{} of the main text]. In Sec.~\ref{sec:scramblingtimeanddoubleexp}, we derive our main result, the universal bound on the scrambling time $t_s$ [Eq.~\maintexteqtsscrboundmain{} of the main text] by proving the double exponential bound on analytic functions with quantitatively determined parameters [Eqs.~\maintexteqdoubleexpbound-\maintexteqtsuniversalbound{} of the main text]. Finally, in Sec.~\ref{sec:nearlyfastscramblers}, we show that quantum mechanical systems that nearly saturate our bound on the scrambling time formally exist, by discussing explicit examples with a suitable (regularized) density of states.

\section{Sustained scrambling and entanglement entropies}
\label{sec:scramblingandentanglement}
\subsection{The mean return probability bounds the average purity of the evolving basis states}
In this subsection, we show that the mean return probability $P_S(t)$ constrains the generation of entanglement, as stated in Eq.~\maintexteqpurityandps{} of the main text. We consider the average purity on the left hand side of this equation, and split it into contributions from matrix elements in the computational basis:
\begin{align}
    \frac{1}{D_S}\sum_{k=1}^{D_S}\Tr[\hat{\rho}_{k,S}^2(t)] &= \frac{1}{D_S} \sum_{k,k_1,k_2=1}^{D_S} \left\lvert {_S}\langle k_1\rvert \Tr_E\left[\hat{\rho}_k(t)\right]\rvert k_2\rangle_S\right\rvert^2 \nonumber \\
    &\geq \frac{1}{D_S} \sum_{k,k_1=1}^{D_S}\left\lvert {_S}\langle k_1\rvert \Tr_E\left[\hat{\rho}_k(t)\right]\rvert k_1\rangle_S\right\rvert^2 \nonumber \\
    &\geq \frac{1}{D_S} \sum_{k=1}^{D_S}\left\lvert {_S}\langle k\rvert \Tr_E\left[\hat{\rho}_k(t)\right]\rvert k\rangle_S\right\rvert^2.
    \label{eqs:diagonalinequalities}
\end{align}
We have dropped the off-diagonal matrix elements $k_1\neq k_2$ and used the non-negativity of the corresponding terms in the second line. In the third line, we have further dropped the $k\neq k_1$ contributions, which are also non-negative, to obtain an inequality focusing on the 3D diagonal $k=k_1=k_2$ contributions. Though a simple technique, dropping off-diagonal terms in such inequalities has been considerably useful in other physical contexts as well~\cite{Zelditch, dynamicalqspeedlimit}.

We can rewrite the last expression in Eq.~\eqref{eqs:diagonalinequalities} in terms of the projectors $\proj_{k} = \lvert k\rangle_S\langle k\rvert \otimes \idop_E$, which gives
\begin{equation}
    \frac{1}{D_S}\sum_{k=1}^{D_S}\Tr[\hat{\rho}_{k,S}^2(t)] \geq \frac{1}{D_S}\sum_{k=1}^{D_S} \Tr\left[\hat{\rho}_k(t)\proj_k\right]^2.
\end{equation}
The right hand side is the mean (over all values of $k$) of the squares of the return probabilities $P_k(t) \equiv \Tr\left[\hat{\rho}_k(t)\proj_k\right] \geq 0$, which must be at least the square of their mean. Thus, we obtain:
\begin{equation}
    \frac{1}{D_S}\sum_{k=1}^{D_S}\Tr[\hat{\rho}_{k,S}^2(t)] \geq \frac{1}{D_S}\sum_{k=1}^{D_S} P_k^2(t) \geq \left(\frac{1}{D_S}\sum_{k=1}^{D_S}P_k(t)\right)^2 = \left[P_S(t)\right]^2,
    \label{eqs:purity_and_ps}
\end{equation}
which is Eq.~\maintexteqpurityandps{} in the main text.

\subsection{The scrambling time \texorpdfstring{$t_s$}{ts} in terms of the scrambled R\'{e}nyi entanglement entropy}

Here, we express the scrambling bound of Eq.~\maintexteqtsscrboundmain{} in terms of the entanglement entropy of a scrambled state, deriving Eq.~\maintexteqtsentropyform{} of the main text.

The second R\'{e}nyi entanglement entropy of the state $\hat{\rho}_k(t)$ in the subsystem $\mathcal{H}_S$ is defined as
\begin{equation}
    S_2[\hat{\rho}_{k,S}^2(t)] \equiv -\ln \Tr\left[\hat{\rho}_{k,S}^2(t)\right].
\end{equation}
This is a direct measure of entanglement~\cite{NielsenChuang}. For example, for pure states, volume-law entanglement corresponds to $S_2 \sim \min\lbrace N_S \ln 2, N_E \ln 2\rbrace$ to leading order, usually associated with maximal scrambling as we also have the general bound:
\begin{equation}
 S_2 \leq \lbrace N_S \ln 2, N_E \ln 2\rbrace
 \label{eqs:entropygeneralbound}
\end{equation}
We now introduce the scrambled entropy cutoff $S_{2,S}$ as the minimum necessary value of the entanglement entropy at which a state is considered scrambled in $\mathcal{H}_S$. In particular, we require that every scrambled state satisfies (as in Eq.~\maintexteqentropyscrambling{} of the main text):
\begin{equation}
    S_2[\hat{\rho}_{k,S}^2(t > t_s)] \geq S_{2,S}
\end{equation}
(where $t<T_D$ is implicit to avoid quantum recurrences~\cite{QuantumRecurrences}). It follows that the mean purity on the left hand side of Eq.~\eqref{eqs:purity_and_ps} is at most $\exp[-S_{2,S}]$ for $t>t_s$. Then Eq.~\eqref{eqs:purity_and_ps} [or Eq.~\maintexteqpurityandps{} in the main text] implies
\begin{equation}
    P_S(t > t_s) \leq \exp\left(-\frac{1}{2} S_{2,S}\right),
    \label{eqs:ps_and_entropy}
\end{equation}
corresponding to replacing $P_{\scr} \to \exp(-S_{2,S}/2)$ in Eq.~\maintexteqscrtimecriterion{} of the main text.
We can substitute this in the bound on the sustained scrambling time $t_s$, given by Eq.~\maintexteqtsscrboundmain{} of the main text, which yields
\begin{equation}
    t_s \geq \frac{c\beta}{\pi}\ln \left[\frac{S_{2,S}+4\ln(f_{\beta}Z_{\max})}{2\expratemain_{\text{eff}}}\right].
    \label{eqs:ts_entropy_main}
\end{equation}
In a thermodynamic limit with $\lvert \ln (f_{\beta}Z_{\max})\rvert = o[S_{2,S}]$ and $\expratemain_{\text{eff}} = \Theta(1)$ (requiring $Z_{\beta}(\tau) = O(1)$ in the interval $\tau \in [\tau_1,\tau_2]$), we get the leading order inequality
\begin{equation}
    t_s \gtrsim \frac{c\beta}{\pi} \ln S_{2,S},
    \label{eqs:ts_entropy_thermo}
\end{equation}
which is Eq.~\maintexteqtsentropyform{} of the main text. In this thermodynamic regime, we can qualitatively phrase our fast scrambling result in simple terms: the scrambling time is at least logarithmic in the scrambled ($2$nd R\'{e}nyi) entanglement entropy.
While Eqs.~\eqref{eqs:ts_entropy_main} and \eqref{eqs:ts_entropy_thermo} directly capture the time required to generate a given degree of entanglement in our setup, we note two reasons for preferring the expression in terms of the scrambled return probability $P_{\scr}$ in Eq.~\maintexteqtsscrboundmain{} of the main text.
\begin{enumerate}
    \item Most significantly, the bound in terms of $P_{\scr}$ can be made tighter. Intuitively, this is because $P_S(t)$ can continue to track aspects of quantum dynamics and scrambling that the entanglement entropy is insensitive to; for example, $S_{2,S}$ is subject to the bound of Eq.~\eqref{eqs:entropygeneralbound}, while $P_{\scr}$ can be taken to be as small as desired. For concreteness, consider maximal scrambling (e.g., to Haar random behavior), for which $P_{\scr} \sim D_S^{-1}$ and $S_{2,S} \sim \min\lbrace N_S,N_E\rbrace \ln 2$. When $N_S \leq N_E$, both give identical bounds to leading order, $t_s \gtrsim (c\beta/\pi)\ln N_S$. But for $N_S > N_E$, while $-\ln P_{\scr} \sim N_S$ continues to hold and constrains $t_s \gtrsim (c\beta/\pi)\ln N_S$ from Eq.~\maintexteqtsscrboundmain{} of the main text, the leading contribution to the entropy comes from $\mathcal{H}_E$ with $S_{2,S} \sim N_E \ln 2$, due to which Eq.~\eqref{eqs:ts_entropy_main} sets a much weaker bound $t_s \gtrsim (c\beta/\pi)\ln N_E$. For a physical example where this can be relevant, see Ref.~\cite{dynamicalqspeedlimit}, where scrambling is considered for subsystem sizes including $N_S > N_E$ (e.g., in the SYK-$4$ model).
    \item $P_{\scr}$ requires only projective measurements in a computational basis with the $D_S$ projectors $\proj_k$, and directly determines the dynamics of observables diagonal in this basis. In contrast, $S_{2,S}$ in principle probes more of the full structure of the reduced density matrices $\hat{\rho}_{k,S}$ in the Hilbert space $H_S$. In particular, it may be dominated by contributions (such as from the off-diagonal matrix elements of $\hat{\rho}_{k,S}$) not relevant to a specific observable of interest.
\end{enumerate}

On the other hand, $S_{2,S}$ can be efficiently measured in experiments, for instance, using local randomized measurements~\cite{AndreasPurityRM}; this means that Eqs.~\eqref{eqs:ts_entropy_main} and \eqref{eqs:ts_entropy_thermo} are also testable experimentally, in addition to the bound in terms of $P_S(t)$.

\section{A quantum speed limit for arbitrary initial states \texorpdfstring{$\hat{\rho}_{\beta E}$}{pbE}}
\label{sec:psktgeneralization}
In this section, we describe the passage from the infinite temperature speed limit of Ref.~\cite{dynamicalqspeedlimit} to the inequality $P_S(t) \geq K_{\beta}(t)$ for arbitrary initial states under Hamiltonian evolution with more details than in the main text, beginning with a brief review of the former. This is essentially an expanded version of the derivation containing Eq.~\maintexteqtimeevolvingstateKraus{} through Eq.~\maintexteqregSFFdef{} in the main text.

\subsection{Derivation of \texorpdfstring{$P_S(t) \geq K_{\beta}(t)$}{Ps(t) >= Kb(t)}}

Consider a general time-dependent completely positive quantum operation~\cite{NielsenChuang} acting on an initial reference state $\hat{\rho}_0$ in $\mathcal{H}$, with any set of $M$ time-independent Kraus operators $\lbrace \kraus{r}{t}\rbrace_{r=1}^{M}$,
\begin{equation}
    \hat{\rho}(t) = \sum_{r=1}^{M} \kraus{r}{t} \hat{\rho}_0\krausd{r}{t}.
\end{equation}
Note that the initial state $\hat{\rho}(0)$ is not necessarily equal to the reference state $\hat{\rho}_0$, and further that the quantum operation is not required to be trace preserving~\cite{NielsenChuang}. Quantum dynamics of this type may be characterized by a generalized SFF,
\begin{equation}
    K(t) \equiv \frac{1}{D^2} \sum_{r=1}^{M} \left\lvert \Tr \kraus{r}{t} \right\rvert^2
\end{equation}
Given this setup, the mean return probability for specific initial reference states $\hat{\rho}_{0,k} = \proj_k/D_E$ (or equivalently, the states of Eq.~\maintexteqinitstatedef{} in the main text with $\hat{\rho}_{\beta E} = \hat{\rho}_{0 E} \equiv \idop_E/D_E$),
\begin{equation}
    P_S(t)[\proj_k] \equiv \frac{1}{D} \sum_{k=1}^{D_S}\Tr[\proj_k(t)\proj_k(0)] = \frac{1}{D} \sum_{k=1}^{D_S}\sum_{r=1}^{M} \Tr[\kraus{r}{t} \proj_k(0) \krausd{r}{t} \proj_k(0)].
    \label{eqs:psdef_analogous_memories}
\end{equation}
was shown to be constrained by~\cite{dynamicalqspeedlimit}
\begin{equation}
    P_S(t)[\proj_k] \geq K(t),
    \label{eqs:dynamicalinequlityinfinitetemp}
\end{equation}
for \textit{any} complete, orthonormal choice of projectors $\proj_k$.

For the above speed limit, it is crucial that the initial reference states $\proj_k/D_E$ form a complete, orthogonal set for the full Hilbert space $\mathcal{H}$. It is this completeness that allows constraining the basis-dependent $P_S(t)$ with the basis-independent $K(t)$ that involves a trace over the entire space $\mathcal{H}$. However, the initial states $\hat{\rho}_{k}(0) = \lvert k\rangle_S\langle k\rvert \otimes \hat{\rho}_{\beta E}$ of Eq.~\maintexteqinitstatedef{} in the main text are complete only in $\mathcal{H}_S$, and generally not in $\mathcal{H}_E$ except in the specific case of the infinite temperature state $\hat{\rho}_{0 E}$.

The resolution to this difficulty comes from writing the initial state $\hat{\rho}_k(0)$ in terms of nontrivial Kraus operators $\kraus{r}{0}$ acting on the initial reference state $\hat{\rho}_{0,k} = \proj_k/D_E$. This is enabled by noting that $\hat{\rho}_{\beta E}$ is a positive operator (has non-negative eigenvalues)~\cite{NielsenChuang}, and therefore admits a positive Hermitian square root $\hat{\rho}_{\beta E}^{1/2}$ (with non-negative eigenvalues). In particular, we have
\begin{equation}
    \hat{\rho}_k(0) = \hat{\rho}^{1/2}_{\beta E} \left(\lvert k\rangle_S \langle k\rvert \otimes \idop_E\right) \hat{\rho}^{1/2}_{\beta E} = \left[D_E^{1/2} \hat{\rho}^{1/2}_{\beta E}\right]\hat{\rho}_{0,k}  \left[D_E^{1/2}\hat{\rho}^{1/2}_{\beta E}\right]^\dagger.
\end{equation}
Now, we can set $M=1$ with $\kraus{1}{0} = D_E^{1/2} \hat{\rho}_{\beta E}^{1/2}$, i.e., a single nonvanishing Kraus operator. As subsequent time evolution is generated by a Hamiltonian $\hat{H}$, we have the time-dependent Kraus operators:
\begin{equation}
    \kraus{r}{t} = e^{-i\hat{H}t} \kraus{r}{0} = D_E^{1/2} e^{-i\hat{H}t}\hat{\rho}_{\beta E}^{1/2} \delta_{r,1}.
\end{equation}
Substituting these Kraus operators in Eq.~\eqref{eqs:psdef_analogous_memories}, we obtain precisely the mean return probability in Eq.~\maintexteqnewscrpsdef{} of the main text:
\begin{equation}
    P_S(t)[\proj_k] = \frac{1}{D}\sum_{k=1}^{D_S}D_E \Tr[e^{-i\hat{H}t}\hat{\rho}_{\beta E}^{1/2} \proj_k \hat{\rho}_{\beta E}^{1/2} e^{i\hat{H} t} \proj_k ] = P_S(t).
\end{equation}
Further, the SFF for these Kraus operators is
\begin{equation}
    K(t) = \frac{1}{D^2} \left\lvert \Tr \kraus{1}{t}\right\rvert^2 = \frac{D_E}{D^2} \left\lvert \Tr\left(e^{-i\hat{H}t} \hat{\rho}_{\beta E}^{1/2}\right)\right\rvert^2 \equiv K_{\beta}(t).
\end{equation}
Now, Eq.~\eqref{eqs:dynamicalinequlityinfinitetemp} for these Kraus operators gives $P_S(t) \geq K_{\beta}(t)$, which is Eq.~\maintexteqdynamicalinequalityagain{} in the main text.

\subsection{Criteria for nontriviality}

We should emphasize that, unlike the $\beta=0$ case where $P_S(0) = K(0) = 1$ under Hamiltonian dynamics, for $\beta \neq 0$ their values at $t=0$ can be different. In particular, $P_S(0) = 1$, while $K_{\beta}(0) = f_{\beta}^2 \leq 1$ [see Eqs.~\maintexteqfidelityinfdef{} and \maintexteqdosscramblingcriteriongen{} of the main text]. Further, as $K_{\beta}(t) \leq K_{\beta}(0)$, we can only obtain a nontrivial bound on the scrambling time from $P_S(t) \geq K_{\beta}(t)$ if
\begin{equation}
f_{\beta}^2 > P_{\scr}.
\label{eqs:nontrivialitycriterion1}
\end{equation}

In this subsection, we ask under what conditions a system may satisfy Eq.~\eqref{eqs:nontrivialitycriterion1}. These considerations are equivalent to, and expand on, the discussion following Eq.~\maintexteqdosscramblingcriteriongen{} of the main text. The need to consider such criteria is in contrast to the $\beta = 0$ case of Ref.~\cite{dynamicalqspeedlimit}, where as long as $P_{\scr} < 1$, some nontrivial $>0$ bound on $t_s$ always exists.

However, the tradeoff in our case still admits two kinds of universality:
\begin{enumerate}
    \item For a given initial state with a certain $f_{\beta}^2$, we get nontrivial values when $P_{\scr} < f_{\beta}^2$. This can translate to a restriction on the system size of $\mathcal{H}_S$. For example, if we are interested in maximal scrambling with $P_{\scr} \sim D_S^{-1}$, the number of qubits $N_S$ in the subsystem $\mathcal{H}_S$ must be as large as
    \begin{equation}
        N_S > 2 \log_2 \frac{1}{f_{\beta}},
    \end{equation}
    for a nontrivial scrambling time. An interesting special case is when the initial state is a pure state, $\rho_{\beta E} = \lvert \ell\rangle_E\langle \ell\rvert$, which has $f_{\beta} = D_E^{-1/2}$. Then, we get $N_S > N_E$ for a nontrivial bound for pure states.
    For any state that is not entirely pure, it follows that we get nontrivial scrambling time bounds even for some values of $N_S \leq N_E$.
    \item For a given scrambling value $P_{\scr} < 1$, the condition $f_{\beta}^2 > P_{\scr}$ restricts the entanglement of the initial state $\hat{\rho}_{\beta E}$. For example, this is because the $1/2$-order R\'{e}nyi entanglement entropy~\cite{HorodeckiEntanglementReview} of this state is determined by the fidelity:
    \begin{equation}
        S_{1/2,E} \equiv 2 \ln \left[D_E^{1/2}f_{\beta}\right].
    \end{equation}
    Thus, for a nontrivial bound, $S_{1/2,E}  > \ln\left[D_E P_{\scr}\right]$, implying that a smaller value of $P_{\scr}$ requires less entanglement in the initial state for our bound to be useful.
\end{enumerate}

\section{Derivation of a universal bound on the scrambling time}
\label{sec:scramblingtimeanddoubleexp}
In this section, we derive a quantitative refinement of a known asymptotic double exponential bound on bounded analytic functions which is originally given in terms of finite but undetermined constants~\cite{AhlforsHeins, Hayman, LinnikOstrovskii, Hayman_v2}. This refinement provides explicit values for the undetermined constants in the previous versions. As mentioned in the main text, this allows us to constrain the scrambling time even for finite dimensional systems and obtain concrete numerical bounds for $t_s$ that can potentially be tested in experiments, instead of just an asymptotic estimate in a thermodynamic limit. This fills in the technical details for the discussion around Eqs.~\maintexteqdoubleexpbound{} through \maintexteqtsuniversalbound{} of the main text. 

Sec.~\ref{sec:theorem_and_ts} states our quantitative version of the double-exponential bound in Theorem~\ref{thm:newdoubleexp}, and derives the scrambling time as per Eqs.~\maintexteqtsscrboundmain{} and \maintexteqtsuniversalbound{} of the main text. Sec.~\ref{sec:newdoubleexpderivation} proves the double exponential bound by conformally mapping the strip to a half-plane, based on two lemmas on analytic functions on the half-plane, closely following the methods of Ref.~\cite{Hayman_v2}.

As a prelude to our technical discussion, we provide some local power series intuition for why analyticity should constrain the decay of a bounded function. For simplicity, let an analytic $F(t)$ have a saddle point at $F(t=0) = 1$. On account of analyticity, we can write a power series expansion in a neighborhood of this point:
\begin{equation}
    F(t - i\tau) = 1 - \frac{1}{2}c_2(t-i\tau)^2 + O((t-i\tau)^3),
\end{equation}
where $c_2 = -F''(0)$. We take $0 < c_2 \in \mathbb{R}$, so that $F(i\tau) \in \mathbb{R}$ [to $O(\tau^2)$] grows as $1+c_2\tau^2/2$ away from $t=0$ while $F(t) \in \mathbb{R}$ [to $O(t^2)$] decays as $1-c_2t^2/2$; in addition, $F(t)$ then locally satisfies $\lvert F(t-i\tau)\rvert < F_R(-i\tau)$ to $O((t-i\tau)^2)$ [like in Eq.~\maintexteqridgeproperty{} of the main text] so it is sufficient to consider $F(-i\tau)$ for boundedness. The key restriction from analyticity here is that the same coefficient $c_2$ determines both the growth and decay. Thus, if $F(-i\tau) < F_{\max}$ in $\tau \in (0, \tau_2)$, we have
\begin{equation}
    F(t) \geq 1-(F_{\max}-1)\left(\frac{t}{\tau_2}\right)^2+O(t^3).
\end{equation}
In other words, the decay rate is bounded by $F_{\max}$. This qualitatively connects analyticity and boundedness to (local) decay rates in an intuitive example, but we emphasize that the mathematical details of our (global) double exponential bound [Eq.~\eqref{eqs:tildoseffdoubleexpthm}, or Eq.~\maintexteqdoubleexpbound{} in the main text] and method of proof are quite different.
Mathematically inclined readers may find
closer intuition to our result through
Jensen's formula from complex analysis or the Poisson kernel formula for harmonic functions \cite{Koosis}.

\subsection{The scrambling time from a double-exponential bound on analytic functions}
\label{sec:theorem_and_ts}

\subsubsection{A quantitative theorem: analytic functions are bounded by a double-exponential in time}

The key to deriving our quantitative bound on the scrambling time is the following theorem bounding the decay of analytic functions, which quantifies the undetermined parameters in the previously known double-exponential bounds of Refs.~\cite{AhlforsHeins, Hayman, LinnikOstrovskii, Hayman_v2} (the latter being stated in Eq.~\maintexteqdoubleexpbound{} of the main text):
\begin{theorem}[\textbf{Quantitative decay rate of an analytic function on a strip}]
\label{thm:newdoubleexp}
Let the function $F(\tcomp)$ be analytic with $\lvert F(\tcomp)\rvert < 1$ in the open strip $\lbrace \tcomp = t-i\tau \in \mathbb{C}$: $\tau \in (\tau_1,\tau_2)\rbrace$,  and continuous in the corresponding closed strip $\tau \in [\tau_1,\tau_2]$. We also require that $F(\tcomp)$ is not identically $0$ in the strip. Then for any $0 < \ell < \infty$, there is an ``exceptional set'' $\texcept \subset [0,\infty)$ of times $t$ whose length is at most
\begin{equation}
    \int_{\texcept}\diff t \leq \frac{(\tau_2-\tau_1)}{\pi}\ell,
    \label{eqs:thmexceptionalsetlength}
\end{equation}
such that $\lvert F(\tcomp)\rvert$ is at least a double exponential in $t$ everywhere in the strip except when $t$ is in the exceptional set, i.e., 
\begin{align}
    \lvert F(t-i\tau)\rvert &\geq \exp\left[-\frac{\Lambda_{\ell}}{2} \exp\left(\frac{\pi t}{\tau_2-\tau_1}\right)\right], \nonumber \\
    &\ \text{for all } t \in [0,\infty) \setminus \texcept, \text{ and } \tau \in [\tau_1,\tau_2]. \label{eqs:tildoseffdoubleexpthm}
\end{align}
Here, the parameter $ \expratemain_{\ell} \geq 0$ is given in terms of $\ell$ and $F(-i\tau)$ by $\expratemain_{\ell} = 2(\pi^2+8) (2+\ell)\exprate/\ell$, in which 
\begin{equation}
    \exprate = \inf_{\tau\in(\tau_1,\tau_2)} \left\lbrace\sec \left[\frac{\pi}{\tau_2-\tau_1} \left(\tau-\frac{\tau_2+\tau_1}{2}\right)\right]\ln \frac{1}{\lvert F(-i\tau)\rvert}\right\rbrace.
    \label{eqs:expratedef}
\end{equation}
\end{theorem}
\begin{proof}
 See Sec.~\ref{sec:newdoubleexpderivation}. The proof closely follows the techniques used to prove Theorem 7.32 in Ref.~\cite{Hayman_v2}, which is related to the present theorem, but unlike Eq.~\eqref{eqs:tildoseffdoubleexpthm} does not quantitatively determine the constants in Eq.~\maintexteqdoubleexpbound{} of the main text.
\end{proof}

Note that if we set $\ell = 0$, the expression in the exponent of Eq.~\eqref{eqs:tildoseffdoubleexpthm} diverges, so we can never guarantee that the exceptional set vanishes for the double exponential form in this equation. However, we can make it as small as desired.

\subsubsection{Applying the bound to \texorpdfstring{$\tildoseff(t)$}{N\~{}b(t)}}

Here, we derive Eqs.~\maintexteqdoubleexpbound{} through \maintexteqexprateresult{} of the main text from Theorem~\ref{thm:newdoubleexp}. As noted in Eq.~\maintexteqtildoseffdef{} of the main text, $\tildoseff(\tcomp)$ is a weighted sum of exponentials with non-negative weights,
\begin{equation}
    \tildoseff(t-i\tau) = \frac{D_E^{1/2}}{f_{\beta}D}\sum_{n=1}^{D} \langle E_n\rvert\hat{\rho}_{\beta E}^{1/2}\lvert E_n\rangle e^{-iE_n t-E_n \tau},
    \label{eqs:tildoseffdef}
\end{equation}
and is consequently analytic and bounded on any strip $t\in\mathbb{R}$, $\tau \in (\tau_1,\tau_2)$, as well as continuous up to the boundary of the strip. As we want to focus on evolution in real time, we take $\tau=0$ to lie in the strip, $\tau_1 \leq 0 \leq \tau_2$. Further, $\tildoseff(0) = 1$ making $\tildoseff(\tcomp)$ not identically $0$. Thus, $\tildoseff(\tcomp)$ already satisfies most of the properties required by Theorem~\ref{thm:newdoubleexp} for $F(\tcomp)$, except not necessarily that $\lvert F(\tcomp)\rvert < 1$ in the interior of the strip.

To normalize $\tildoseff(\tcomp)$ to satisfy this last property, we consider the maximum of $\tildoseff(t)$ in the closed strip $\tau \in [\tau_1,\tau_2]$. From Eq.~\eqref{eqs:tildoseffdef} we have $Z_{\beta}(\tau) \equiv \tildoseff(-i\tau) \in [0,\infty)$ and $\lvert \tildoseff(t-i\tau)\rvert \leq Z_{\beta}(\tau)$ [Eq.~\maintexteqridgeproperty{} of the main text], which amount to the ``ridge property''~\cite{LinnikOstrovskii} of Fourier transforms of probability distributions --- implying that the maximum of $\tildoseff(\tcomp)$ occurs on the imaginary axis $t=0$.

Further, $Z_{\beta}(\tau)$ is also a convex function of $\tau$ (see, e.g., Ref.~\cite{LinnikOstrovskii}); in our case, this is because
\begin{equation}
    \frac{\diff^2 Z_{\beta}}{\diff \tau^2}(\tau) = \frac{D_E^{1/2}}{f_{\beta}D}\sum_{n} E_n^2 \langle E_n\rvert\hat{\rho}_{\beta E}^{1/2}\lvert E_n\rangle e^{-E_n \tau} \geq 0.
    \label{eqs:Zconvexity}
\end{equation}

Here, there are two distinct cases of interest:
\begin{enumerate}
    \item All energy levels with $\langle E_m\rvert\hat{\rho}_{\beta E}^{1/2}\lvert E_m\rangle \neq 0$ have $E_m = 0$. In this (trivial) case, $\tildoseff(\tcomp) = 1$ everywhere, automatically satisfying the double exponential bound [Eq.~\maintexteqdoubleexpbound{} of the main text] for any choice of decay rate $\Lambda \geq 0$.
    \item Otherwise, the partition function is strictly convex, $\diff^2 Z_{\beta}(\tau)/\diff \tau^2 > 0$ in $[\tau_1,\tau_2]$ (as all terms in Eq.~\eqref{eqs:Zconvexity} are nonnegative, and not all of them vanish). In this case, the maximum
    \begin{equation}
        Z_{\max} \equiv \max_{\tau \in [\tau_1,\tau_2]} Z_{\beta}(\tau)
    \end{equation}
    is attained exclusively on either $\tau=\tau_1$ or $\tau=\tau_2$, with $Z_{\beta}(\tau) < Z_{\max}$ in $\tau \in (\tau_1,\tau_2)$. Then, we can define
    \begin{equation}
        F(\tcomp) \equiv \frac{\tildoseff(\tcomp)}{Z_{\max}}, \text{ satisfying } \lvert F(t-i\tau)\rvert < 1 \text{ in } \tau \in (\tau_1,\tau_2).
        \label{eqs:FunctionAndTildoseff}
    \end{equation}
    Theorem~\ref{thm:newdoubleexp} now applies to this $F(\tcomp)$, and we obtain the quantitative double exponential bound,
    \begin{equation}
        \lvert \tildoseff(t)\rvert \geq Z_{\max} \exp\left[-\frac{\expratemain_{\ell}}{2}\exp\left(\frac{\pi t}{\tau_2-\tau_1}\right)\right], \text { for } t\geq 0: t \notin \texcept.
    \end{equation}
    corresponding to Eqs.~\maintexteqdoubleexpbound, \maintexteqexceptionalsetlength{} and \maintexteqexprateresult{} of the main text (respectively from Eqs.~\eqref{eqs:tildoseffdoubleexpthm}, \eqref{eqs:thmexceptionalsetlength} and \eqref{eqs:expratedef} of the Theorem, with Eq.~\maintexteqexprateresult{} obtained by extending the right hand side of Eq.~\eqref{eqs:expratedef} to be continuous at $\tau_1$, $\tau_2$ to replace the infimum with a minimum over $[\tau_1,\tau_2]$).
\end{enumerate}

A key point to emphasize is that these results specify the previously undetermined constants such as $\Lambda_{\ell}$ entirely in terms of the behavior of $\tildoseff(t-i\tau)$ at the initial (real) time $t=0$, or equivalently, the partition function $Z_{\beta}(\tau) = \tildoseff(-i\tau)$. This is due to taking the physical viewpoint that we should not expect to know aspects of the late-time dynamics of the system [e.g., $\tildoseff(t>0)$] when constraining this very dynamics. From a purely mathematical standpoint, we can potentially make $\exprate$ smaller (i.e., obtain a tighter bound) in some cases if we already know some aspects of the $t>0$ behavior 
(by extending Lemma~\ref{lem:improvedhayman} in Sec.~\ref{sec:newdoubleexpderivation} to $\lvert z\rvert \geq 1$; see also \cite{inpreparation}).

Finally, to provide some physical intuition for the decay rate $\exprate$, we can carry out the minimization in Eq.~\eqref{eqs:expratedef} by finding a stationary point $\tau_0$ (if it exists). Differentiation yields the implicit equation:
\begin{equation}
    \frac{\pi}{\tau_2-\tau_1} \tan\left[\frac{\pi}{\tau_2-\tau_1} \left(\tau_0-\frac{\tau_2+\tau_1}{2}\right)\right] \ln \frac{Z_{\max}}{Z_{\beta}(\tau_0)} = \left.\frac{\diff \ln Z_{\beta}(\tau)}{\diff \tau}\right\rvert_{\tau=\tau_0}.
\end{equation}
If the stationary point does not exist, the minimum occurs at $\tau_{\max} \in \lbrace \tau_1,\tau_2\rbrace$ where $Z_{\beta}(\tau_{\max}) = Z_{\max}$. Whether $\tau_0$ exists or not, the value of $\exprate$ is in all cases conveniently expressed, using familiar thermodynamic relations~\cite{LLStatMech, ReichlStatMech}, in terms of the expectation value of the energy $\mathcal{E}_{\beta}(\tau) \equiv -\diff \ln Z_{\beta}(\tau)/\diff\tau$ and the excess free energy $\Delta \mathcal{F}_{\beta}(\tau) \equiv -\tau^{-1}[\ln Z_{\beta}(\tau)-\ln Z_{\max}]$ at Euclidean time $\tau$ [with $\Delta \mathcal{F}_{\beta}(\tau_{\max}) = 0$]. We get:
\begin{equation}
    \exprate = \sqrt{\tau_0^2 [\Delta \mathcal{F}_{\beta}(\tau_0)]^2 + \frac{(\tau_2-\tau_1)^2}{\pi^2}[\mathcal{E}_{\beta}(\tau_0)]^2}.
\end{equation}
This expresses the decay rate in terms of thermodynamic quantities. The location of the point $\tau_0$ and the significance of the (free) energy at this point are system-dependent, and it appears that we cannot obtain further universal insights from this expression.

\subsubsection{Derivation of the scrambling time}

We can obtain a bound on the scrambling time $t_s$ by using Eqs.~\eqref{eqs:tildoseffdoubleexpthm} and \eqref{eqs:FunctionAndTildoseff} in Eq.~\maintexteqdosscramblingcriteriongen{} of the main text, which states that (implicitly with $t < T_D$, the time scale of quantum recurrences~\cite{QuantumRecurrences})
\begin{equation}
\left\lvert \tildoseff(t > t_s)\right\rvert^2 \leq \frac{P_{\scr}}{f_{\beta}^2},
    \label{eqs:dosscramblingcriterion_gen} 
\end{equation}
However, we should carefully account for the exceptional set $\texcept$. Any part of the exceptional set in the region $t\geq t_s$ does not affect the scrambling time, as it refers to when $\lvert \tildoseff(t)\rvert$ itself already satisfies our scrambling criterion and whether it violates Eq.~\eqref{eqs:tildoseffdoubleexpthm} becomes immaterial. However, the length of the portion of the exceptional set before $t_s$ is uncertain: it can be $0$ at minimum, or as much as the maximum allowed length $(\tau_2-\tau_1)\ell/\pi$ of the exceptional set. So we can only constrain the scrambling time to within this maximum possible length. This corresponds to requiring $\lvert \tildoseff(t_s+(\tau_2-\tau_1)\ell/\pi)\rvert^2$ to be less than the scrambling value. Assuming $0 \in [\tau_1,\tau_2]$ and using Eq.~\eqref{eqs:tildoseffdoubleexpthm} with \eqref{eqs:FunctionAndTildoseff}, we get the inequality
\begin{equation}
    Z_{\max}\exp\left[-\frac{(2+\ell)}{\ell}(\pi^2+8) \exprate \exp\left(\frac{\pi t_s}{\tau_2-\tau_1}+\ell\right)\right] \leq \frac{P_{\scr}^{1/2}}{f_{\beta}}.
\end{equation} 
Rearranging this equation, we obtain
\begin{equation}
    t_s \geq \frac{\tau_2-\tau_1}{\pi}\left\lbrace\ln\left[\frac{1}{2}\ln\left(\frac{f_{\beta}^2 Z_{\max}^2}{P_{\scr}}\right)\right]-\ln\left[\frac{(2+\ell)}{\ell}(\pi^2+8) \exprate\right]-\ell\right\rbrace.
    \label{eqs:ts_general_form_supplement}
\end{equation}
This leads to Eq.~\maintexteqtsuniversalbound{} of the main text.
Since the choice of $\ell$ is arbitrary, it is best to choose a value that maximizes the bound for $t_s$ (though this is not crucial in a thermodynamic limit, where this term is usually subleading). The corresponding maximization of $\ln[\ell/(2+\ell)]-\ell$ over $\ell > 0$ gives
\begin{equation}
    \ell =  \sqrt{3}-1.
\end{equation}
Using this value of $\ell$ in Eq.~\eqref{eqs:ts_general_form_supplement} yields \maintexteqtsscrboundmain{} of the main text.

\subsection{Proof of \texorpdfstring{Theorem~\ref{thm:newdoubleexp}}{Theorem 1} on the quantitative decay of analytic functions}
\label{sec:newdoubleexpderivation}

\subsubsection{Conformal map from a strip to a half-plane}

In order to prove Theorem~\ref{thm:newdoubleexp}, it will be convenient to work on the right \emph{half-plane} instead of the strip, since several formulas from complex analysis are simpler on the half-plane. There is a conformal map between the two regions
shown in Fig.~\ref{fig:map} and defined below, which can translate results on the right half-plane to results on the strip $\tau_1<-\im\tcomp<\tau_2$.
\begin{figure}[!htb]
\centering


\includegraphics{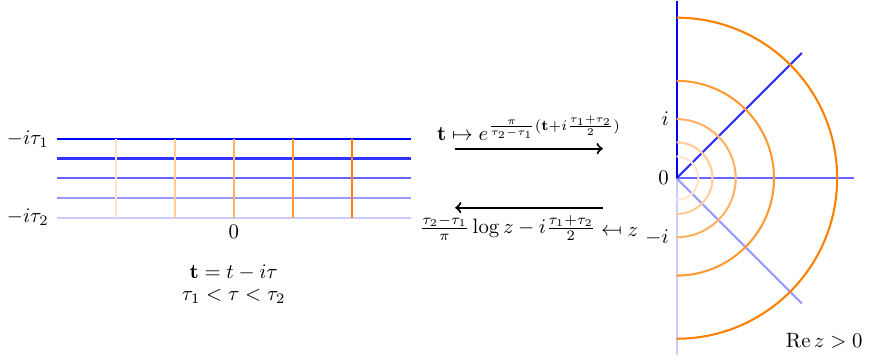}
\caption{Conformal map between the strip $\{\tcomp=t-i\tau:\tau_1<\tau<\tau_2\}$ and right half-plane $\re z>0$, with the boundaries included. Horizontal lines in the strip are mapped to radial rays $re^{i\theta}$, $r>0$, in the half-plane.
}\label{fig:map}
\end{figure}

Specifically, we define the variables ($z\in \mathbb{C}$; $r,\theta,x,y \in \mathbb{R}$) via the transformation
\begin{equation}
    z = re^{i\theta} = x+iy = \exp\left[\pi \frac{\tcomp+i\tfrac{\tau_2+\tau_1}{2}}{\tau_2-\tau_1}\right].
\end{equation}
Then for $(-\im \tcomp) \in [\tau_1,\tau_2]$,
\begin{align}
    r &= \exp\left[\frac{\pi \re \tcomp}{\tau_2-\tau_1}\right] \in (0,\infty), \qquad
    \theta = \pi \frac{\im \tcomp + \tfrac{\tau_2+\tau_1}{2}}{\tau_2-\tau_1} \in \left[-\frac{\pi}{2},\frac{\pi}{2}\right],
    \label{eqs:rtheta}
\end{align}
which is seen to correspond to the closed half-plane $\{z=x+iy:x\ge0,y\in\R\}$ 
(with $z=0$ corresponding to values $-\infty-i\tau$ in the strip). We will freely switch between these different coordinates for the right half-plane in what follows.

\subsubsection{Constraint on regions of double-exponential or faster decay in \texorpdfstring{$t$}{t}}

On account of Eq.~\eqref{eqs:rtheta}, a double exponential decay in $t$ corresponds to an exponential decay in $r$. Our goal is now to constrain the size of regions $r \in \rexcept_\exprateparam$ in which an analytic function $F(z)$ may decay faster than the exponential rate $\exp(-\exprateparam r)$.
This is accomplished using the following Lemma, which is an immediate consequence of Theorem 7.27 in Ref.~\cite{Hayman_v2}:
\begin{lemma}[\textbf{Constraining regions of exponential or faster decay}]
\label{lem:hayman727}
Let $F(z)$ satisfying $\lvert F(z)\rvert < 1$ be an analytic function that is not identically $0$ in the open half-plane $x>0$. For any $\exprateparam > 0$, let $\rexcept_{\exprateparam} \subset (0, \infty)$ be a measurable set such that
\begin{equation}
    \inf_{\theta \in \left(-\frac{\pi}{2},\frac{\pi}{2}\right)}\lvert F(r e^{i\theta})\rvert \leq e^{-\exprateparam r},\ 
    \;\; \text{for all } r\in V_{\exprateparam}.
    \label{eqs:expboundandset}
\end{equation}
The set $V_{\exprateparam}$ consists of radii $r$ in which $\lvert F(r e^{i\theta})\rvert$ is at most $e^{-\exprateparam r}$ for some value of $\theta$ in the right-half-plane. Then, if we define the logarithmic length of $\rexcept_{\exprateparam}$ in $r$ (corresponding in the strip to its actual length in $t$, up to constants) by
\begin{equation}
    L[\rexcept_{\exprateparam}] \equiv \int_{V_{\exprateparam}} \frac{\diff r}{r},
    \label{eqs:rexceptlength}
\end{equation}
the following bound is satisfied everywhere on the half-plane $x>0$:
\begin{equation}
    \ln \frac{1}{\lvert F(z)\rvert} \geq \frac{4}{(2+L[\rexcept_{\exprateparam}])(\pi^2+8)}\int_{\rexcept_{\exprateparam}} \diff r' \frac{x}{x^2 + (r'+\lvert y\rvert)^2} \exprateparam r'.
    \label{eq:Hayman_uz_inequality}
\end{equation}
\end{lemma}
\begin{proof}
    This follows from setting $u(z) = \ln \lvert F(z)\rvert$ and $f(r) = \exprateparam r$ in Theorem 7.27 of Ref.~\cite{Hayman_v2} (see also Theorem 7.32 of Ref.~\cite{Hayman_v2}).
\end{proof}

To provide some relation to more familiar expressions, we note that the integral on the right hand side is related to electrostatics (i.e., solutions of the Laplace equation) on the half-plane~\cite{AhlforsHeins, Hayman_v2}. In particular,
the Poisson kernel formula~\cite{Koosis, Jackson} gives the following solution to Laplace's equation for the electrostatic potential:
\begin{equation}\label{eq:phi}
    \phi(x,y) = \frac{1}{\pi}\int_{-\infty}^{\infty} \diff y' \frac{x}{x^2 + (y'-y)^2} \phi(0,y'),
\end{equation}
which we will use with the boundary condition $\phi(0,y) = \exprateparam y$ for $y \in \rexcept_{\exprateparam}$ and $0$ otherwise.
With this context, Eq.~\eqref{eq:Hayman_uz_inequality} can be expressed for $y \leq 0$ as follows~\cite{Hayman_v2}:
\begin{equation}
    \ln \frac{1}{\lvert F(x+iy)\rvert} \geq \frac{4\pi}{(2+L[\rexcept_{\exprateparam}])(\pi^2+8)}\phi(x,y).
    \label{eq:hayman_uz_electrostatic}
\end{equation}
The right hand side should be considered~\cite{AhlforsHeins} an indirect measure of the size of the set $\rexcept_{\exprateparam}$.

We make two quick remarks concerning the mathematical context of the above lemma. First, in Ref.~\cite{Hayman_v2}, the result is stated more generally for {subharmonic} functions $u$, which includes $u(z)=\ln \lvert F(z)\rvert$ for  analytic $F(z)$ as a special case. 
The subharmonicity suggests why one might expect to see a Poisson kernel-like formula in
Eq.~\eqref{eq:Hayman_uz_inequality}.
Second, the actual Poisson kernel formula (see e.g. Ref.~\cite[\S IIIG.2]{Koosis}) for $\ln|F(z)|$ can be used in place of Lemma~\ref{lem:hayman727} to obtain quantitative decay rate estimates like in Theorem~\ref{thm:newdoubleexp}, but only on the \emph{boundary} of the strip. 
In \cite{inpreparation}, we explore the utility of these Poisson kernel boundary estimates in other situations.
However, in the present setting, since we will often work in the \emph{interior} of a strip, using Lemma~\ref{lem:hayman727}, which provides decay rate estimates on the entirety of the strip, will produce better $t_s$ bounds whenever $\tau_1,\tau_2\ne0$. This is because larger strips, which are not restricted to have a boundary at $\tau=0$, will increase the prefactor $c\beta$ measuring the strip width in our bound on $t_s$; in the appropriate thermodynamic limit [Eq.~\maintexteqtsentropyform{} of the main text or \eqref{eqs:ts_entropy_thermo} here], this is the \textit{only} parameter available to improve the tightness of the bound [provided that $Z_\beta(\tau) = O(1)$ continues to hold in the larger strip].

\subsubsection{Direct constraint on the decay parameter $\exprateparam$}

Lemma~\ref{lem:hayman727} can be used to derive the following Lemma (closely related to Theorem 7.32 in Ref.~\cite{Hayman_v2}):

\begin{lemma}[\textbf{Constraining the decay parameter}]
\label{lem:improvedhayman}
Let $F(z)$ be as in Lemma~\ref{lem:hayman727}. For any $\lambda > 0$ and $R>1$, let $\rexcept_{\exprateparam} \subseteq [1,R]$ be a set of exponential-or-faster decay satisfying Eq.~\eqref{eqs:expboundandset}, and therefore subject to Lemma~\ref{lem:hayman727}.
Then the
logarithmic length $L[\rexcept_\exprateparam] \leq \ln R < \infty$ of $\rexcept_{\exprateparam}$ is related to the decay rate $\exprateparam$ by
\begin{equation}
    \exprateparam \leq \frac{(2+L[\rexcept_{\exprateparam}])}{L[\rexcept_{\exprateparam}]}(\pi^2+8)\left\lbrace\frac{1}{\cos\theta}\ln \frac{1}{\lvert F(e^{i\theta})\rvert}\right\rbrace,
    \label{eqs:exceptionalsetlengthinequality1}
\end{equation}
where $z = e^{i\theta}:$ $\theta \in (-\pi/2, \pi/2)$ is any point on the unit semi-circle in the right half-plane (which corresponds to the vertical segment
$t=0$ in the original strip).
\end{lemma}
\begin{proof}
Our proof will closely follow that of Theorem 7.32 in Ref.~\cite{Hayman_v2}.  In Eq.~\eqref{eq:Hayman_uz_inequality}, we set $z=e^{i\theta}$, and specifically consider only these points in the integral:
\begin{equation}
    \int_{\rexcept_{\exprateparam}} \diff r' \frac{x}{x^2 + (r'+\lvert y\rvert)^2} \exprateparam r'.
\end{equation}
In this integral, $\lvert x\rvert, \lvert y\rvert \leq 1$ due to the restriction to $\lvert z\rvert=1$, and further $r' \geq 1$ as $\rexcept_{\exprateparam} \subseteq [1,\infty)$ by assumption; this implies $\lvert z\rvert, \lvert y\rvert \leq r'$. The denominator thus satisfies the inequality
\begin{equation}
    x^2 + (r'+\lvert y\rvert)^2 = \lvert z\rvert^2+(r')^2 + 2(r')\lvert y\rvert \leq 4(r')^2.
\end{equation}
Using this inequality in Eq.~\eqref{eq:Hayman_uz_inequality} with $x=\cos\theta$ gives
\begin{equation}
    \ln \frac{1}{\lvert F(e^{i\theta})\rvert} \geq \frac{\lambda \cos \theta}{(2+L[\rexcept_{\exprateparam}])(\pi^2+8)}\int_{\rexcept_{\exprateparam}} \frac{\diff r'}{r'}.
\end{equation}
Rearranging factors and noting that the integral on the right hand side equals $L[\rexcept_{\exprateparam}]$ by Eq.~\eqref{eqs:rexceptlength}, we get Eq.~\eqref{eqs:exceptionalsetlengthinequality1}.
\end{proof}

The above result differs from Theorem 7.32 of Ref.~\cite{Hayman_v2}, which determines the limit of $r^{-1}\ln \lvert F(r e^{i\theta})\rvert$ as $r\to\infty$, outside an exceptional set of some (unspecified) finite logarithmic length. However, Eq.~\eqref{eqs:exceptionalsetlengthinequality1} provides an explicit quantitative constraint on the logarithmic length of the exceptional set $\rexcept_\exprateparam$ pertaining to a specific $\exprateparam$. 
As an aside, we note that Theorem 7.32 of Ref.~\cite{Hayman_v2} has connections with the Phragm\'{e}n-Lindel\"{o}f principle~\cite{AhlforsHeins, Hayman, Hayman_v2} (which generalizes the maximum modulus principle to e.g. unbounded domains). 
This principle also plays a role in the conjectured bound on chaos in out-of-time-ordered correlators (OTOCs)~\cite{MSSotocBound}, albeit in a quantitatively very different form.

\subsubsection{Completing the proof}

Only a few quick observations remain to obtain Theorem~\ref{thm:newdoubleexp} from Eq.~\eqref{eqs:exceptionalsetlengthinequality1}. We note that the function $(2+\ell)/\ell$ is monotonically decreasing with $\ell$. Further, as $F(\tcomp)$ is analytic and not identically zero in the open strip by assumption, the identity theorem~\cite{SteinShakarchi2, BrownChurchill} implies that $F(e^{i\theta})$ is not identically zero in $\theta \in (-\pi/2,\pi/2)$. Thus, for any given $0 < \ell < \infty$, if we choose
\begin{equation}
    \exprateparam = \frac{(2+\ell)}{\ell}(\pi^2+8)\left\lbrace\frac{1}{\cos\theta}\ln \frac{1}{\lvert F(e^{i\theta})\rvert}\right\rbrace
    \label{eqs:exprateparamchoice}
\end{equation}
for any $\theta \in (-\pi/2,\pi/2)$ such that $F(e^{i\theta}) \neq 0$, then it follows from Eq.~\eqref{eqs:exceptionalsetlengthinequality1} that
\begin{equation}
    \frac{2+\ell}{\ell} \leq \frac{(2+L[\rexcept_{\exprateparam}])}{L[\rexcept_{\exprateparam}]},\ \;\;\text{ which implies } L[\rexcept_{\exprateparam}] \leq \ell.
    \label{eqs:lengthboundproof}
\end{equation}
Eq.~\eqref{eqs:lengthboundproof} holds for any value of $R$, so we can take $R\to\infty$ in Lemma~\ref{lem:improvedhayman}. This implies that given any $\ell > 0$, by choosing $\exprateparam$ according to Eq.~\eqref{eqs:exprateparamchoice}, we can guarantee that any set $\rexcept_{\exprateparam} \in [1,\infty)$ in which $\lvert F(r e^{i\theta})\rvert \leq \exp(-\exprateparam r)$ for some $\theta$ must have a logarithmic length no greater than $\ell$.

Given that the choice of $\theta$ in Eq.~\eqref{eqs:exprateparamchoice} is arbitrary, we will get a smaller value of $\lambda$, and therefore a tighter exponential in Eq.~\eqref{eqs:expboundandset}, by minimizing the right hand side of Eq.~\eqref{eqs:exprateparamchoice} with respect to $\theta$, while Eq.~\eqref{eqs:lengthboundproof} continues to hold. At this stage, transforming back to the strip $\tcomp = t-i\tau$ gives (as the unit semi-circle corresponds to the vertical $t=0$ segment in the strip)
\begin{equation}
    \lambda = \frac{(2+\ell)}{\ell}(\pi^2+8)\inf_{\tau\in(\tau_1,\tau_2)} \left\lbrace\sec \left[\frac{\pi}{\tau_2-\tau_1} \left(\tau-\frac{\tau_2+\tau_1}{2}\right)\right]\ln \frac{1}{\lvert F(-i\tau)\rvert}\right\rbrace,
    \label{eqs:expparaminterior}
\end{equation}
as an appropriate choice of $\exprateparam$ so that any set $\texcept \in [0,\infty)$ (the transformed version of $\rexcept_{\exprateparam}$) such that
\begin{align}
    \inf_{\tau \in (\tau_1,\tau_2)}\lvert F(t-i\tau)\rvert \leq \exp\left[-\exprateparam \exp\left(\frac{\pi t}{\tau_2-\tau_1}\right)\right],\ \;\; \text{for all } t \in \texcept,
    \label{eqs:expboundandset2}
\end{align}
must have a length of at most
\begin{equation}
    \int_{\texcept}\diff t \leq \frac{(\tau_2-\tau_1)}{\pi}\ell.
    \label{eqs:lengthrestriction2}
\end{equation}
Finally, as $F(\tcomp)$ is continuous up to the boundary of the strip, Eq.~\eqref{eqs:expboundandset2} can also be extended to $\tau = \tau_1,\tau_2$. This completes the proof of Theorem~\ref{thm:newdoubleexp}.
\qed

\section{Formal examples of nearly fast scramblers}
\label{sec:nearlyfastscramblers}
In this section, we construct formal examples of systems that nearly saturate our bound $t_s \gtrsim (c\beta/\pi)\ln S_{2,S}$ on fast scrambling in the thermodynamic limit. Specifically, we construct the \textit{regularized} density of states of these systems at an assumed inverse temperature $\beta$, assuming that a suitable initial state $\hat{\rho}_{\beta E}$ exists, and show that with sufficiently random eigenstates $t_s \lesssim (2c\beta/\pi)\ln S_{2,S}$ for these systems. This corresponds to the discussion of ``nearly fast scramblers'' after Eq.~\maintexteqtsuniversalbound{} in the main text.

\subsection{Regularized density of states and its Fourier transform}

In this section, we construct simple examples of continuum limits of $\tildoseff(t)$ that attain a double exponential decay rate for $t\in\R$, subject to the requirement that its inverse Fourier transform (the regularized density of states) satisfies $\doseff(E)\ge0$ and $\int_{-\infty}^\infty\doseff(E)\,dE=1$. 
(See also \cite{Ostrovskii}, \cite[p.35, Appendix II]{LinnikOstrovskii} for a different class of examples.) These examples demonstrate that the double exponential decay rate in 
Eq.~\maintexteqdoubleexpbound{} of the main text is essentially optimal (up to a constant in the exponent), even with the constraint that $\doseff(E)\ge0$.
In \cite{inpreparation}, we also formally verify that this leads to similar scrambling time bounds for corresponding finite dimensional systems. 

First consider the analytic function 
\begin{align}\label{eqs:g-def}
g(\tcomp)=\exp(-e^{\frac{\pi}{\cbeta}\tcomp})\exp(-e^{-\frac{\pi}{\cbeta}\tcomp}),\end{align} for $\tcomp=t-i\tau\in\C$ and a fixed $\cbeta>0$, which has absolute value
\begin{align}\label{eqs:g-norm}
|g(\tcomp)| &= \exp\left(-e^{\frac{\pi t}{\cbeta}}\cos\Big(\frac{\pi \tau}{\cbeta}\Big)\right)\exp\left(-e^{-\frac{\pi t}{\cbeta}}\cos\Big(\frac{\pi \tau}{\cbeta}\Big)\right).
\end{align}
The function $g$ is thus bounded in the strip $|\im \tcomp|\le \cbeta/2$, but does not meet the requirement that its inverse Fourier transform is nonnegative. (One can either check this numerically, or note that $g$ is everywhere analytic, but not maximized at $t=0$ when $\tau=\cbeta$, and so even its normalized version cannot be the Fourier transform of a nonnegative density \cite[\S II.3]{LinnikOstrovskii}.)

However, recall the convolution of two functions, 
$g_1*g_2(t) \equiv \int_{-\infty}^\infty g_1(y)g_2(t-y)\,dy,$
behaves with the inverse Fourier transform $g^\vee(E)\equiv\frac{1}{2\pi}\int_{-\infty}^\infty g(t)e^{itE}\,dt$ as
\begin{align}\label{eqn:conv}
(g_1*g_2)^\vee(E) &= (2\pi) g_1^\vee(E)g_2^\vee(E).
\end{align}
Therefore the normalized self-convolution 
\begin{align}\tildoseff(t)=\frac{g*g(t)}{\|g\|_2^2},\;\;\text{where }\; \|g\|_2^2=\int_{-\infty}^\infty g(t)^2\,dt,\label{eqs:nearlyfastscrambler}\end{align} has nonnegative inverse Fourier transform 
\begin{align}
\doseff(E)=\frac{2\pi}{\|g\|_2^2}(g^\vee(E))^2\ge 0,
\end{align}
which is also properly normalized since $\int_{-\infty}^\infty \doseff(x)\,dx=\tildoseff(0)=1$.
Additionally, $\tildoseff(t)$ has similar decay rate on the real line as $g$:
\begin{align*}
\tildoseff(t)= \int_{-\infty}^\infty g(y)g(t-y)\, dy &\le \bigg(\int_{|y|\ge t/2} + \int_{|y-t|\ge t/2}\bigg) [g(y)g(t-y)]\, dy\\\numberthis\label{eqn:fg-ineq}
&\le 2g(t/2) \int_{-\infty}^\infty g(y)\,dy,
\end{align*}
so that
$
|\tildoseff(t)| \le C g(t/2),
$
where $C=2\int_{-\infty}^\infty g(t)\,dt$ is independent of $t$.
The rapid decay of $g$ for $|\im \tcomp|<\cbeta/2$ also ensures that $\tildoseff(\tcomp)$ is defined, bounded, and integrable along horizontal lines in any smaller strip $|\im\tcomp|\le \cbeta/2-\epsilon$.

The double exponential decay rate $\propto\exp(-e^{\frac{\pi}{2\cbeta}|t|})$ can thus be achieved by $\tildoseff(t)$ analytic and bounded in a strip $|\im \tcomp|<\cbeta/2-\epsilon$ for any $\varepsilon>0$, 
with $\doseff(E)\ge0$ and normalized. 
Verifying the conditions in e.g. \cite[Theorem IX.14]{ReedSimon2} shows the density of states $\doseff(E)$ satisfies an exponential decay bound.
Plots demonstrating the decay of these $\tildoseff(t)$ and $\doseff(E)$ are shown in Fig.~\ref{fig:f-ft}.

Surprisingly, $\doseff(E)$ and $\tildoseff(t)$ look very similar to a Gaussian distribution expected for generic many-body systems~\cite{DeltaStar2,ShenkerThouless}; for $c\beta/\pi = 1$ as in Fig.~\ref{fig:f-ft}, the deviation of the latter from a Gaussian becomes appreciable (more than $10\%$ of the Gaussian) only around $t\approx 1.9$, at which $\lvert \tildoseff(1.9)\rvert^2 \approx 0.014 \in (2^{-7},2^{-6})$. This suggests that the difference in quantum dynamics between the examples constructed here and generic many-body systems (in terms of the contribution from the energy eigenvalues) may become significant only for systems of $N \gg 6$ qubits [so that the $O(D^{-1})$ fluctuations in $\lvert \tildoseff(t)\rvert^2$, which is the square of a sum of $D$ phase factors, are negligible compared to its value at this time].

\begin{figure}[htb]
\centering
\includegraphics[width=\textwidth]{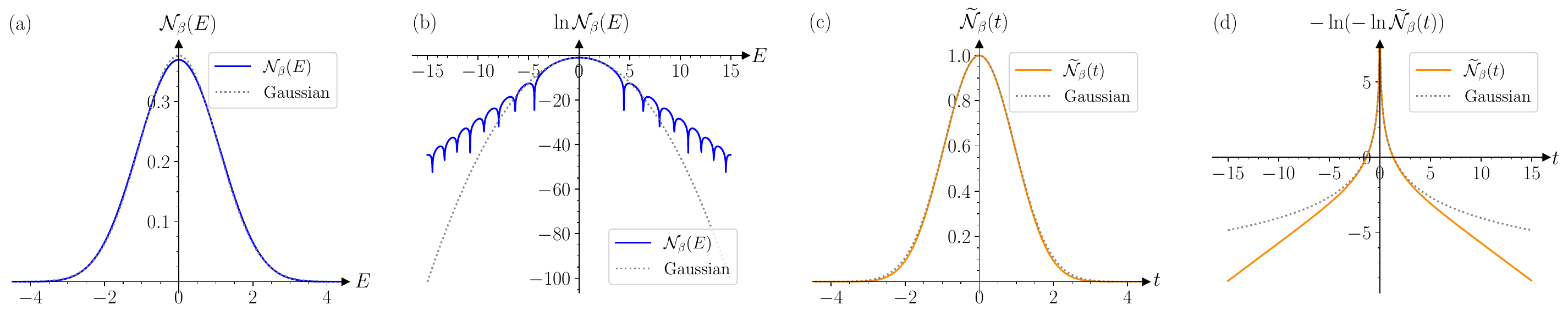}
\caption{(a)-(b) Plots of $\doseff(E)=\frac{2\pi}{\|g\|_2^2}(g^\vee(E))^2$, for $g(t)=\exp(-e^t)\exp(-e^{-t})$ which corresponds to $\cbeta=\pi$ in Eq.~\eqref{eqs:g-def}, compared against a Gaussian probability density with the same variance. 
The spacing between energy samples is $\Delta E=0.025$.
While the plot of $\doseff(E)$ in (a) looks similar to a Gaussian, higher precision integration using the {mpmath} library \cite{mpmath} and a logarithmic scale in (b) demonstrate the exponential (linear in $\ln \doseff(E)$) decay, rather than Gaussian decay. 
(c)-(d) Plots of $\tildoseff(t)=\frac{g*g(t)}{\|g\|_2^2}$ compared to a Gaussian $e^{-t^2\sigma^2/2}$, where $\sigma^2$ is the variance from (a). 
The spacing between time samples is $\Delta t=0.05$.
While the plot of $\tildoseff(t)$ in (c) again looks similar to a Gaussian, higher precision integration and a doubly-logarithmic scale in (d) demonstrate the much faster double-exponential (linear in $-\ln(-\ln \tildoseff(t))$) decay of $\tildoseff(t)$.
} \label{fig:f-ft}
\end{figure}

\subsection{Upper bound on the scrambling time for these examples}

The constraint $\lvert \tildoseff(t)\rvert \leq Cg(t/2)$ leads to a double exponential decay of the regularized SFF for large $t$:
\begin{equation}
    K_{\beta}(t\to\infty) \lesssim f_{\beta}^2 C^2\exp\left[-2\exp\left(\frac{\pi t}{2c\beta}\right)\right]
    \label{eqs:SFFnearlyfastscrambler}
\end{equation}
It is also known that for systems with energy eigenstates that are ``sufficiently random'' [i.e., behave as if sampled from the Haar distribution, including but not restricted to Haar random energy eigenstates in the computational basis], the expectation values of observables $\proj_\alpha$ in initial states $\hat{\rho}_k(0)$ [Eqs.~\maintexteqinitstatedef{} and \maintexteqprojdef{} of the main text] essentially track the SFF to decay to their (maximally) scrambled values~\cite{dynamicalqspeedlimit, Reimann2016, ChaosComplexityRMT, CotlerHunterJones2}:
\begin{equation}
    \Tr[\hat{\rho}_k(t) \proj_\alpha] = \frac{1+O(D_E^{-1/2})}{D_S}+\left\lbrace\Tr[\hat{\rho}_k(0) \proj_\alpha]-\frac{1}{D_S}\right\rbrace K_{\beta}(t).
\end{equation}
Thus, all these expectation values are guaranteed to maximally scramble after a time $t_s$, i.e.,
\begin{equation}
    \Tr[\hat{\rho}_k(t)\proj_\alpha] \sim \frac{1}{D_S},\ \text{ for } t : t_s < t < T_D,
\end{equation}
if $K_{\beta}(t > t_s) = o(D_S^{-1})$~\cite{dynamicalqspeedlimit}, e.g., $K_{\beta}(t > t_s) \leq D_S^{-1-\varepsilon}$ for any $\varepsilon > 0$.  Using Eq.~\eqref{eqs:SFFnearlyfastscrambler} with this condition gives for the scrambling time:
\begin{equation}
    t_s \lesssim \frac{2 c \beta}{\pi} \ln \left[\ln \left(f_{\beta}^2 C^2 D_S^{1+\varepsilon}\right)\right].
\end{equation}
Recalling that $S_{2,S} \sim N_S \ln 2$ for maximal entanglement (for $N_S < N_E$), and further assuming an initial state such that $\lvert \ln (C f_\beta)\rvert = o(N_S)$, we get a leading order scrambling time of
$t_s \lesssim (2 c\beta/\pi)\ln S_{2,S}$.
We conclude that the systems constructed in this section nearly saturate the logarithmic-in-entanglement-entropy bound of Eq.~\eqref{eqs:ts_entropy_thermo} [Eq.~\maintexteqtsentropyform{} in the main text], if the energy eigenbasis is sufficiently random with respect to the computational basis.

\end{document}